\documentclass[runningheads]{llncs}

\usepackage[T1]{fontenc}

\usepackage{amsmath}
\usepackage{amssymb}
\usepackage{mathtools}
\usepackage{verbatim} 
\usepackage{stmaryrd}
\usepackage{url}
\allowdisplaybreaks
\usepackage{basic} 
\usepackage{ambients}
\usepackage{paralist}
\usepackage{graphicx}

\usepackage{amsmath}
\usepackage{amssymb}
\usepackage{mathtools}
\usepackage{verbatim} 
\usepackage{stmaryrd}
\usepackage{setspace}
\usepackage{url}
\usepackage{bbm} 
\allowdisplaybreaks
\usepackage{basic}
\usepackage{ambients}
\usepackage{paralist}
\usepackage{graphicx}
\usepackage{placeins}
\usepackage{hyperref}

\renewcommand{\timeout}[2]{\lfloor #1 \rfloor #2  }
\renewcommand{\tick}{{\scriptstyle \mathsf{tick}}}
\newcommand{\fineC}{{\scriptstyle \mathsf{end}}}

\newif\ifdraft\drafttrue

\newcommand{\ntrans}[1]{\mathrel{{\trans{#1}}\makebox[0em][r]{$\not$\hspace{2ex}}}{\!}}

\newcommand{\rel}{\mathcal R}

\newcommand{\on}{\mathsf{on}}
\newcommand{\off}{\mathsf{off}}

\newcommand{\confCPS}[2]{#1 \, {\Join} \, #2}

\newcommand{\defn}{\triangleq}
\renewcommand{\operatorname}[1]{\mathit{#1}}

\usepackage{marvosym}

\newcommand{\Edit}{\mathsf{E}}

\newcommand{\funEdit}[1]{\big \llbracket{#1} \big \rrbracket}
\renewcommand{\transf}[1]{\big \llbracket{#1} \big \rrbracket}

\newcommand{\ActSet}{\mathsf{Act}}
\newcommand{\SensSet}{\mathsf{Sens}}
\newcommand{\ChanSet}{\mathsf{Chn}}

\usepackage{tikz}
\usetikzlibrary{arrows,positioning,shapes}
\usetikzlibrary{positioning}
\tikzset{
	>=stealth',
	punkt/.style={
		rectangle,
		rounded corners,
		draw=black, very thick,
		text width=6.5em,
		minimum height=2em,
		text centered},
	pil/.style={
		->,
		thick,
		shorten <=2pt,
		shorten >=2pt,}
}
\tikzstyle{line} = [draw, -latex,thick,
shorten <=2pt,
shorten >=2pt]

\newcommand{\Xrec}{\mathsf{X}}
\newcommand{\Yrec}{\mathsf{Y}}

\newcommand{\Sys}{N}
\newcommand{\Malware}{\mathit{M}}
\renewcommand{\confCPS}[2]{#1 \! \vdash \!  {\boldsymbol{ \{ }}#2{\boldsymbol{ \} }}}

\makeatletter
\newcommand{\subalign}[1]{%
	\vcenter{%
		\Let@ \restore@math@cr \default@tag
		\baselineskip\fontdimen10 \scriptfont\tw@
		\advance\baselineskip\fontdimen12 \scriptfont\tw@
		\lineskip\thr@@\fontdimen8 \scriptfont\thr@@
		\lineskiplimit\lineskip
		\ialign{\hfil$\m@th\scriptstyle##$&$\m@th\scriptstyle{}##$\crcr
			#1\crcr
		}%
	}
}
\makeatother

\renewcommand{\on}[1]{ {\mathsf{\scriptstyle on_{#1}}}  }
\renewcommand{\off}[1]{ {\mathsf{\scriptstyle off_{#1}}}  }
\newcommand{\open}{ {\mathsf{\scriptstyle open}}  }
\newcommand{\close}{ {\mathsf{\scriptstyle close}}  }

\newcommand{\req}[1]{ {\mathsf{\scriptstyle turn{#1}}}}

\newcommand{\ctrlDim}[1]{ \mathsf{dim}(#1)  }

\begin{document}  

\title{A process calculus approach to  correctness enforcement  of PLCs (full version)\thanks{An extended abstract will appear in the CEUR Workshop Proceedings of the \emph{21st Italian Conference on Theoretical Computer Science (ICTCS 2020).}}}

\titlerunning{A process calculus approach to enforcement of PLCs} 

\author{Ruggero Lanotte\inst{1} \and 
Massimo Merro\inst{2}  \and 
Andrei Munteanu\inst{2} }

\authorrunning{R. Lanotte and M. Merro and A. Munteanu}
\institute{Universit\`a dell'Insubria, Como, Italy
	\and 
	Universit\`a degli Studi di Verona, Verona,   Italy
}

\pagestyle{plain}%

\maketitle
\begin{abstract}
We define a simple process calculus, based on  Hennessy and Regan's \emph{Timed Process Language}, for specifying networks of communicating  \emph{programmable logic controllers} (PLCs) enriched with  monitors 
 enforcing specification compliance at runtime. We define a synthesis algorithm that given an uncorrupted PLC returns a monitor that enforces the correctness of the PLC, even when  injected with \emph{malware} that may forge/drop actuator commands and  inter-controller communications.    Then, we strengthen  the capabilities of  our monitors by allowing the insertion of actions  to \emph{mitigate}  malware activities. This gives us \emph{deadlock-freedom monitoring}:  malware may not drag   monitored controllers into   deadlock  states. Last but not least, our enforcing monitors represent an effective formal mechanism for prompt detection of malicious activities within PLCs. 
\end{abstract}


\section{Introduction}


 \emph{Industrial Control System} (ICSs) are distributed systems controlling physical processes via \emph{programmable logic controllers} (PLCs) connected to  \emph{sensors} and  \emph{actuators}. PLCs  have an ad-hoc architecture 
 to execute simple processes known as \emph{scan cycles}. Each scan cycle consists  of three phases: (i) reading of the \emph{sensor measurements} of the physical process; (ii) derivation of the commands to guide the evolution of the physical process;
 (iii) transmission of the calculated \emph{commands}  to the \emph{actuator devices}.

 Published scan data show how thousands of  
 PLCs are directly accessible from the Internet \cite{Radvanovsky2013}. 
When this is not the case,   PLCs  are  often connected to each other in  \emph{field communications networks}, opening the way to the spreading of 
worms 
such as the PLC-Blaster worm~\cite{BLACKHAT2016} 
or the PLC PIN Control attack~\cite{Abbasi2016}. 

As a consequence,  extra \emph{trusted hardware components} have been  proposed to enhance the security of  ICS architectures~\cite{McLaughlin-ACSAC2013,Mohan-HiCONS2013}. 
In this respect, McLaughlin~\cite{McLaughlin-ACSAC2013} proposed to add a policy-based \emph{enforcement mechanism}  to mediate  the actuator commands transmitted by the PLC to the physical  plant, whereas Mohan et al.~\cite{Mohan-HiCONS2013} introduced an architecture in which every PLC runs under the scrutiny  of a 
 \emph{monitor} which looks for deviations with respect to \emph{safe behaviours}; 
 if the behaviour of the PLC is not as expected 
 then the control passes to a \emph{safety controller} which maintains the plant within the required safety margins. 

Both  architectures above have been validated by means of simulation-based techniques. However, as far as we know, formal methodologies have not been used yet to model and formally verify security-oriented architectures for ICSs. 

The \emph{goal} of the paper is to verify the effectiveness of a \emph{process calculus approach}
to formalise \emph{runtime enforcement} of specification compliance in  networks \nolinebreak of \nolinebreak PLCs  injected with  \emph{colluding malware} that may forge/drop both actuator commands and inter-controller communications\footnote{We do not deal with alterations of sensor signals  within a PLC, as they can already be  altered either at the network level or within the sensor devices~\cite{ACM-survey2018}.}. 
\emph{Process calculi} represent a successful and widespread formal approach in \emph{concurrency theory} relying on a variety of \emph{behavioural equivalences} (\emph{e.g.}, trace equivalence and bisimilarity) for studying complex systems, such as 
IoT systems~\cite{LBdF13,BDFG17,LaMe18} and cyber-physical systems~\cite{LaMeTi18}, and used in many \nolinebreak areas, including verification of security protocols~\cite{appliedpi,spi} and  security analysis of \emph{cyber-physical attacks}~\cite{LMMV20}. 
On the other hand,  \emph{runtime enforcement} \cite{Schneider2000,Ligatti2005,Falcone3} 
 is a  powerful verification/validation technique 
 aiming at correcting possibly-incorrect executions \nolinebreak of a system-under-scrutiny (SuS) via a kind of monitor 
 that acts as a \emph{proxy}  between the SuS  and its environment. 
%
%

 Thus, we propose  to synthesise
a proxy from an uncorrupted PLC, to form a \emph{monitored PLC}  ensuring:
\begin{itemize}
  \item  \emph{observation-based monitoring}, \emph{i.e.}, the proxy should only look at the observables of the PLC, 
and not at its internal 
 execution; 
\item \emph{transparency}, \emph{i.e.}, the semantics of the monitored PLC must not differ from the semantics of the genuine (\emph{i.e.}, uncorrupted) PLC; 
\item \emph{sound execution}  of the monitored PLC, to prevent incorrect executions;
\item \emph{deadlock-freedom}, \emph{i.e.}, an injected malware may not drag a monitored PLC into a deadlock state;
 \item \emph{prompt detection} of PLC misbehaviours to rise alarms addressed to system engineers (who will reinstall an obfuscated variation of the code  \nolinebreak of
   \nolinebreak  the \nolinebreak PLC);
   \item \emph{mitigation} of malicious activities within the monitored PLC. 
\end{itemize}

Obviously,  if the PLC is compromised then its correct execution can only be enforced with the help of an extra component, a  \emph{secured proxy}, as advocated by McLaughlin~\cite{McLaughlin-ACSAC2013} and  Mohan et al.~\cite{Mohan-HiCONS2013}.
This means that any implementation of our proposed proxy should be \emph{bug-free} to deal with possible infiltrations of \nolinebreak   malware.
This may seem like we just moved the problem over to securing the proxy. However,  this is not the case because the proxy only needs to enforce correctness at runtime, while the PLC controls its physical process relying on malware-prone communications via the Internet or the USB ports. Of course, by no means runtime reconfigurations of the secure proxy should be allowed.

\paragraph{Contribution.}

We define a simple timed process calculus, based on  Hennessy and Regan's \emph{Timed Process Language} (TPL)~\cite{HR95},   for specifying  networks of communicating monitored controllers, possibly injected with \emph{colluding malware} that may forge/drop both actuator commands and inter-controller communications. Monitors are formalised in terms of a sub-class of finite-state  Ligatti  et al.'s  \emph{edit automata}~\cite{Ligatti2005}. 
A  network composed of $n$ PLCs $\mathrm{Ctrl}_i$,  running in parallel, each of which injected  with a malware $\mathrm{Malw}_i$, and enforced by a monitor $\mathsf{Mon}_i$, is represented as:
\begin{displaymath}
{
	\small
\confCPS{\mathsf{Mon}_1}{\mathrm{Ctrl}_1 | \mathrm{Malw}_1} \, 
\parallel
 \ldots
\parallel
\, 
\confCPS{\mathsf{Mon}_n}{\mathrm{Ctrl}_n | \mathrm{Malw}_n} \, . 
}
\end{displaymath}
Here, the parallel process  $\mathrm{Ctrl}_i | \mathrm{Malw}_i$ is a formal abstraction of the sequential execution of the PLC code $\mathrm{Ctrl}_i $ injected with the malware $\mathrm{Malw}_i$.

Then, we propose a \emph{synthesis function} $\funEdit{-}$ that, given an
 uncorrupted deterministic PLC  $\mathrm{Ctrl}$ 
returns, in \emph{polynomial time}, a \emph{syntactically deterministic} \cite{Acetoetal2017} edit automaton $\funEdit{\mathrm{Ctrl}}$  to form a monitored PLC 
that ensures: \emph{observation-based monitoring}, \emph{transparency},   \emph{sound execution} of the monitored PLC,  \emph{prompt detection} of alterations of \nolinebreak the behaviour of the monitored PLC.
All these properties can be expressed with a single algebraic equation: 
  \begin{equation}
    \small 
  \label{eq_intro}
  {
\prod_{i=1}^{n}\confCPS{\funEdit{\mathrm{Ctrl}_i}}{\mathrm{Ctrl}_i | \mathrm{Malw}_i}
\; \simeq \; 
\prod_{i=1}^{n}\confCPS{\go}{\mathrm{Ctrl}_i}
}
\end{equation}
\noindent 
for arbitrary malware $\mathrm{Malw}_i$, where $\simeq$ denotes   \emph{trace equivalence} and $\go$ is the monitor that allows any action.
Here, intuitively, each monitor $\funEdit{\mathrm{Ctrl}_i}$ prevents incorrect executions of the compromised controller $\mathrm{Ctrl}_i | \mathrm{Malw}_i$. 

However,  
our monitors do not protect against 
 malware that 
   may drag a monitored PLC  into a deadlock state. In fact,  Equation~\ref{eq_intro} does not hold \nolinebreak with respect to \emph{weak bisimilarity}, which is a notoriously  \emph{deadlock-sensitive} semantic equivalence. 
Thus, in order to achieve deadlock-freedom we equip our  monitors 
with the semantic capability to  \emph{mitigate} those malicious activities that  \nolinebreak may deadlock the controller. In practice, our  monitors will be able to \emph{insert} actions, \emph{i.e.}, to emit correct actions in full autonomy to complete  scan cycles. 
The enforcement resulting from the introduction of mitigation allows us to recover \emph{deadlock-freedom monitoring\/} by proving Equation~\ref{eq_intro}
 with respect to weak bisimilarity.

 \paragraph{Outline.} Section~\ref{sec:calculus} defines our process calculus to express monitored  controllers injected with malware.
 Section~\ref{sec:case-study} provides a non-trivial and modular use case  in the context of  \emph{water transmission networks}.
  Section~\ref{sec:safety} defines an algorithm to synthesise  our monitors.
 Section~\ref{sec:liveness} introduces mitigation to recover 
 deadlock-freedom. Section~\ref{sec:conclusion}  draws conclusions and discusses related   work. Full proofs can be found   
in the appendix. 
 
\section{A timed process calculus for monitored PLCs}
\label{sec:calculus}
We define our process calculus 
as an extension of Hennessy and Regan's~TPL~\cite{HR95}. 

Let us start with some preliminary notation. We use $s, s_k \in \mathsf{Sens}$ for \emph{sensor signals},  $a,a_k \in \mathsf{Act}$ for \emph{actuator commands}, and   $c, c_k \in\mathsf{Chn}$ for \emph{channel names}.

\paragraph{Controller.}
In our setting, controllers are nondeterministic sequential  timed processes evolving through three different  phases: \emph{sensing} of  sensor signals, \emph{communication} with other controllers, and \emph{actuation}. For convenience,  we use four different syntactic categories to distinguish the four main states of a controller: $\mathbbm{Ctrl}$  for  initial states, $\mathbbm{Sens}$ for  sensing states, $\mathbbm{Com}$ for  communication states, and $\mathbbm{Act}$ for actuation states. In its initial state, a controller is a recursive process starting its scan cycle in the \emph{sensing phase}: 
\begin{displaymath}
\begin{array}{rcl} 
\mathbbm{Ctrl} \ni P & \Bdf & \fix \Xrec S 
\end{array}
\end{displaymath}
Notice that due to the cyclic behaviour of controllers, the process variable $\Xrec$ may syntactically occur only in 
the last phase, actuation.  We assume \emph{time guarded} recursion to avoid undesired \emph{zeno behaviours}. Intuitively, in time guarded recursion the process variable  must  occur  prefixed by at least one timed action \nolinebreak $\tick$.

During the sensing phase, the controller waits  for a \emph{finite} number of admissible sensor signals. If none of those signals arrives in the current time slot then the controller will \emph{timeout} moving to the following time slot (we adopt the TPL construct $\timeout{\cdot}{\cdot}$ for timeout). The controller may also sleep for a while, waiting for sensor signals to become stable. The syntax is the following: 
\begin{displaymath}
\begin{array}{rcl}
\mathbbm{Sens} \ni S   & \Bdf & \timeout{\sum_{i \in I} s_i.S_i}{S} 
  \Bor \tick.S 
 \Bor C 
\end{array}
\end{displaymath}

\begin{table}[t]
  \caption{LTS for controllers }
  \centering
          {\small
\begin{math}
\begin{array}{l@{\hspace*{8mm}}l}
\Txiom{Rec}
{ S{\subst {\fix \Xrec S} \Xrec} \trans{\alpha}  {S'} }  
{  \fix \Xrec S \trans{\alpha}  {S'}}
&
\Txiom{TimeS}
{-}
{ \tick.S  \trans{\tick}   S}
\\[10pt]
\Txiom{ReadS}
{j \in I}  
{ \timeout{\sum_{i\in I} s_i.S_i}{S} \trans{s_j}  {S_j}}
&
\Txiom{TimeoutS}
{ - }  
{ \timeout{\sum_{i\in I} s_i.S_i}{S} \trans{\tick}  {S}}
\\[10pt]
\Txiom{InC}
{j \in I}  
{ \timeout{\sum_{i\in I} c_i.C_i}{C} \trans{c_j}  {C_j}}

&
\Txiom{TimeoutInC}
{- }  
{ \timeout{\sum_{i\in I} c_i.C_i}{C} \trans{\tick}  {C} 
}
\\[10pt]
\Txiom{OutC}
{ - }  
{\timeout{\overline{c}.C}{C'} \trans{\overline{c}}  {C}}
&
\Txiom{TimeoutOutC}
{- }  
{\timeout{\overline{c}.C}{C'} \trans{\tick}  {C'}
}
\\[10pt]
\Txiom{WriteA}
{-}
{ { \overline{a}.A  \trans{\overline{a}}   A}}
&
\Txiom{End}
{-}
{ { \fineC.P \trans{\fineC}  P}}
\end{array}
\end{math}
}
\label{tab:sem-ctrl}
\end{table}

Once the sensing phase is concluded, the controller starts its \emph{calculations} that may depend  on \emph{communications} with other controllers.  
Controllers communicate to each other 
for mainly two reasons: either to receive notice about the state of other physical sub-processes or to require an actuation on a different physical process that will have an  influence on the  physical process governed by the controller. 
 We adopt a \emph{channel-based} \emph{handshake point-to-point}  communication paradigm. As PLCs usually work under timing constraints, our communication is always under timeout. The syntax for the communications phase is: 
\begin{displaymath}
\begin{array}{rcl}
\mathbbm{Comm} \ni C & \Bdf & \timeout{\sum_{i \in I}c_i.C_i}{C}\Bor \timeout{\overline{c}.C}{C} \Bor A 
\end{array}
\end{displaymath}
Thus, our controllers can either listen on a \emph{finite} number of communication channels  or  transmit on specific channels to pass some local information.

Finally, in the \emph{actuation phase} the controller eventually transmits a \emph{finite} sequence of commands to a number of different actuators, and then, it emits a special signal $\fineC$ 
 to denote the end of the scan  cycle. 
  After that, it restarts its cycle in the sensing phase via a recursive call denoted with  a process variable \nolinebreak $\Xrec$.
In order to ensure semantics closure, we also  have a construct $ \fineC.P$ which will be only generated at runtime but never  used to write  PLC programs.
\begin{displaymath}
\begin{array}{rcl}
\mathbbm{Act} \ni A  & \Bdf & \overline{a}.A  \Bor  \fineC.\Xrec \Bor   \fineC.P
\end{array}
\end{displaymath}


\begin{remark}[Scan cycle duration and maximum cycle limit]
  \label{rem:maximum-time}
  Notice that any scan cycle of a PLC must be completed within a \emph{maximum cycle limit} which depends on the controlled physical process; if this time limit is violated  the PLC stops and throws an exception~\cite{BLACKHAT2016}.
         Thus, the signal $\fineC$ must occur well before the \emph{maximum cycle limit}. We assume that our PLCs successfully complete their scan cycle in
          less than half of the maximum cycle limit.   
\end{remark}

The operational semantics of  controllers is given in Table~\ref{tab:sem-ctrl}. This is very much along the lines of Hennessy and Regan's TPL~\cite{HR95}.  In the following, we use the metavariables $\alpha$ and $\beta$ to range over the set of possible actions: {\small $\{ s, \overline{a}, a , \overline{c}, c,  \tau, \tick, \fineC \}  $}. These actions denote:  sensor  readings, actuator commands, drops of actuator commands, channel transmissions, channel receptions/drops,   internal actions,  passage of time, and  end of a scan cycle, respectively.

\paragraph{Malware.}
Let us  provide  a formalisation of the malware code that we assume may be injected in  a controller to compromise its runtime behaviour. The kind of malware we wish to deal with may perform the following malicious activities:
\begin{itemize}
\item forging fake channel transmissions towards other controllers (via actions $\overline{c}$); 
\item dropping incoming  communications from other controllers (via actions $c$); 
\item forging fake actuator commands (via actions $\overline{a}$); 
\item dropping  actuator commands launched by the controller (via actions $a$). 
\end{itemize}

The formal syntax of the admitted malware is the following: 
\begin{displaymath}
\begin{array}{rcl}
\mathbbm{Malw} \ni M & \Bdf & \timeout{\sum_{i \in I}\mu_i.M_i}{M}  \Bor \fix \Xrec M \Bor \Xrec \Bor \tick.M \Bor \nil
\end{array}
\end{displaymath}
where  the prefixes $\mu_i \in \{  \overline{c}, c , \overline{a},  a  \}$, for $i \in I$, denote the possible malicious actions mentioned above. Again, we assume \emph{time guarded} recursion to avoid undesired \emph{zeno behaviours}  introduced by the malware, that is, a malware can never  prevent the passage of time \emph{ad infinitum} in a controller\footnote{In general,  malware that aims to take control of the plant  has no interest in delaying the scan cycle and risking the violation of the maximum cycle limit  whose consequence would be the immediate controller shutting down~\cite{BLACKHAT2016}.}.

A straightforward operational semantics is given in Table~\ref{tab:sem-malware}.

\begin{table}[t]
  \caption{LTS for malware code}
  \centering 
{\small
\begin{math}
\begin{array}{c}

\Txiom{Malware}
{j \in I}  
{ \timeout{\sum_{i \in I}\mu_i.M_i}{M} \trans{\mu_j}  {M_j}}
\Q\Q
\Txiom{TimeoutM}
{ - }  
{ \timeout{\sum_{i \in I}\mu_i.M_i}{M} \trans{\tick}  {M}}
\\[15pt]
\Txiom{RecM}
{ \Malware {\subst {\fix \Xrec M}{\Xrec}} \trans{\alpha}  {M'} }  
{  \fix \Xrec M \trans{\alpha}  {M'}}
\Q
\Txiom{TimeM}
{-}
{ \tick.M  \trans{\tick}   M}
\Q
\Txiom{TimeNil}
{-}
{ \nil  \trans{\tick}   \nil}
\end{array}
\end{math}
}
\label{tab:sem-malware}
\end{table}

\paragraph{Compromised controller.}

In our setting, a compromised controller is  a controller that may potentially run in parallel with an arbitrary piece of malware. The syntax is the following: 
\[
\begin{array}{rcl}
Z & \Bdf & P \Bor S \Bor C \Bor A \\[2pt]
 \mathbbm{CCrtl} \ni J & \Bdf & Z \Bor Z |  M
\end{array}
\]
where $Z \in \mathbbm{Ctrl} \cup \mathbbm{Sens} \cup \mathbbm{Comm} \cup \mathbbm{Act} $ denotes a controller in an arbitrary state, and  $|$ is the standard process algebra construct for parallel composition.

The operational semantics of a compromised controller is given by the transition rules of Table~\ref{tab:comp-Ctrl}. 
Rule \rulename{Ctrl} models the genuine behaviour of the controller even in the presence of the malware (possibly waiting for a proper trigger). 
Rule \rulename{Inject} denotes the injection of a malicious action fabricated by the malware. Rule \rulename{DropAct} models the drop of an actuator command $\overline{a}$; in this manner, the command $\overline{a}$ never reaches its intended actuator device. 
Rule \rulename{TimePar} models \emph{time synchronisation} between the controller and the malware (we recall that malware cannot exhibit zeno behaviours).

\begin{remark}[Attacks on channels]
Notice that injection/drop on communication channels affects the interaction between controllers and not within them. For this reason, we do not have a rule for channels similar to 
\rulename{DropAct}. Inter-controller malicious activities on communication channels will be prevented by the \nolinebreak monitor. 
\end{remark}

\begin{table}[t]
  \caption{LTS for compromised controllers}
  \centering 
{\small 
\begin{math}
\begin{array}{l@{\hspace*{5mm}}l}
\Txiom{Ctrl}
{ Z \trans{\alpha} Z' \Q \alpha \neq \tick}  
{ Z | \Malware \trans{\alpha} Z' | \Malware}
&
\Txiom{Inject}
{  \Malware \trans{\alpha} \Malware' \Q \alpha \not \in \{ \tick, a \}}  
{ Z | \Malware \trans{\alpha} Z | \Malware'}
\\[10pt]
\Txiom{DropAct}
{Z \trans{\overline a} Z' \q \Malware \trans{a} \Malware'}  
{ Z | \Malware \trans{\tau} Z' | \Malware'}
&
{\Txiom{TimePar}
{Z \trans{\tick} Z' \q \Malware \trans{\tick} \Malware' \q  
}  
{ Z | \Malware \trans{\tick} Z' | \Malware'}}
\end{array}
\end{math}
}
\label{tab:comp-Ctrl}
\end{table}

\paragraph{Monitored controller(s).}
The core of our runtime enforcement relies on a  (timed) sub-class of finite-state  Ligatti et al.'s \emph{edit automata}~\cite{Ligatti2005}, \emph{i.e.}, a particular class of automata specifically designed to modify/suppress/insert actions in a generic system in order to preserve its correct behaviour.
Their syntax follows: 
\begin{displaymath}
\begin{array}{rcl}
\mathbbm{Edit} \ni \Edit & \Bdf &  \go  \Bor  
\sum_{i \in I} \eact{\alpha_i}{\beta_i}.\Edit_i \Bor \fix \Xrec \Edit \Bor \Xrec 
\end{array}
\end{displaymath}
Intuitively, the automaton $\go$ will admit any action of the monitored system, while the edit automaton $\sum_{i \in I} \eact{\alpha_i}{\beta_i}.\Edit_i$ \emph{replaces} actions $\alpha_i$  with $\beta_i$, and then continues as $\Edit_i$, for any $i \in I$, with $I$ finite. The operational semantics of our  edit automata is the following:
\begin{center}
{\small 
\begin{math}
\begin{array}{c}

\Txiom{Go}
{-}
{ { \go  \trans{\alpha / \alpha}  \go }}

\Q\q
\Txiom{Edit}
{j \in I}
{ \sum_{i \in I} \eact{\alpha_i}{\beta_i}.\Edit_i  \trans{\eact{\alpha_j}{\beta_j}}   \Edit_j }

\Q\q

\Txiom{recE}
{ \Edit {\subst {\fix \Xrec \Edit} {\Xrec} } \trans{\eact{\alpha}{\beta}}  {\Edit'} }  
{  \fix \Xrec \Edit \trans{\eact{\alpha}{\beta}}  {\Edit'}}
\end{array}
\end{math}
}
\end{center}
When an edit automaton performs a transition labeled $\eact{\alpha}{\beta}$, with $\alpha \neq \tau$ and $\beta=\tau$, we say that the automaton \emph{suppresses} the observable action $\alpha$.

Our \emph{monitored controllers}, written $\confCPS{\Edit}{J}$, are constituted by a (potentially) compromised controller $J$ and an edit automaton $\Edit$ enforcing the behaviour of  $J$  according to the following transition rule
for correction/suppression: 

 \begin{displaymath}
{
\Txiom{Enforce}
{  J \trans{\alpha}  J' \Q  \Edit \trans{\eact{\alpha}{\beta}}  \Edit'}
{  \confCPS {\Edit}{J} \trans{\beta}  \confCPS {\Edit'}{J'}} \, . 
}
\end{displaymath}
In a monitored controller  $\confCPS{\Edit}{J}$ with no malware inside, the enforcement never occurs, \emph{i.e.}, in rule \rulename{Enforce} we always have $\alpha=\beta$, and  the two components $\Edit$ and $J$ evolve in a tethered fashion, moving through related correct states.

We can easily generalise the concept of monitored controller to a \emph{field communications network} of parallel monitored controllers, each one acting on different actuators, 
 and  exchanging information via channels. These networks are formally defined via the  grammar:

\[
\mathbbm{FNet} \ni \Sys \Bdf \confCPS{\Edit}{J} \Bor \Sys \parallel \Sys 
\]
and described via the operational semantics given in Table~\ref{tab:sem-net}. Notice that monitored controllers may interact with each other via channel communication. Moreover, they may evolve in time when no communication occurs (we recall that neither controllers nor malware admit zeno behaviours). This ensures us  \emph{maximal progress}~\cite{HR95}, a desirable time property  when modelling real-time systems: communications are never postponed to future time slots. 

\begin{table}[t]
  \caption{LTS for monitored field communications networks}
  \centering 
{\small 
\begin{math}
\begin{array}{c}
\Txiom{ParL}
{ \Sys_1 \trans{\alpha} \Sys_1'}  
{ \Sys_1 \parallel \Sys_2 \trans{\alpha} \Sys_1' \parallel \Sys_2}
\Q\Q
\Txiom{ParR}
{ \Sys_2 \trans{\alpha} \Sys_2'}  
{ \Sys_1 \parallel \Sys_2 \trans{\alpha} \Sys_1 \parallel \Sys_2'}
\\[10pt]

  \Txiom{ChnSync}
{\Sys_1 \trans{c} \Sys_1' \q\, \Sys_2 \trans{\overline{c}} \Sys_2'}  
{ \Sys_1 \parallel \Sys_2 \trans{\tau} \Sys_1' \parallel \Sys_2'}
\\[10pt] 
{\Txiom{TimeSync}
{\Sys_1 \trans{\tick} \Sys_1' \Q \Sys_2 \trans{\tick} \Sys_2' \Q  \Sys_1 \parallel \Sys_2 \ntrans{\tau}}  
{ \Sys_1 \parallel \Sys_2 \trans{\tick} \Sys_1' \parallel \Sys_2'}
}
\end{array}
\end{math}
}
\label{tab:sem-net}
\end{table}

	Having defined operational semantics of  a monitored field network, we can easily concatenate single transitions to define \emph{execution traces}. 
	\begin{definition}[Execution traces]
		Given a trace {\small $t=\alpha_1 \ldots \alpha_k$}, we  write {\small $N \trans{t} N'$} as an abbreviation for {\small $N=N_0\trans{\alpha_1}N_1\trans{\alpha_2} \cdots \trans{\alpha_{k-1}}N_{k-1}\trans{\alpha_k}N_k=N'$}. 
	\end{definition}

	Execution traces  can be used to formally define both notions of   \emph{ano\-ma\-ly} \emph{detection} and \emph{correction},  achieved by the monitoring edit automaton.  Intui\-tively, the  detection occurs whenever the edit automaton does not allow  the  execu\-tion of a certain \emph{observable action} $\alpha$ proposed by a compromised controller; if $\alpha $ 
	is replaced with a different action $\beta$ then the  automaton does: (i)  \emph{correction}, if $\beta \neq \tau$, and (ii)  \emph{suppression}, if $\beta = \tau$. 
	\begin{definition}[Anomaly detection and correction]
		\label{def:detection}
		Let  $\confCPS{\Edit}{J}$ be a monitored controller, with $J= P | M$ being a  compromised controller.  
		We say that  the  edit automaton 
		$\Edit$  \emph{detects an anomaly} of $J$ during the execution of some \emph{observable action} $\alpha$ in the trace $t  \alpha$, only if:
		\begin{itemize}
			\item    {\small $P\trans{t}Z$}  (\emph{i.e.}, $t$ is a genuine trace of $P$); 
			\item  {\small $J\trans{t\alpha}J'$}, for some $J'\!$; 
			\item    {\small $\confCPS{\Edit}{J} \trans{t}   N $}, for some $N$ (\emph{i.e.}, $\Edit$ does allow  the trace $t$), and   {\small $\confCPS{\Edit}{J}  \trans{t\alpha}  N  $}, for no $N$  (\emph{i.e.}, $\Edit$ does not allow the trace $t \alpha$).
		\end{itemize}
		We say that  $\Edit$ \emph{corrects} (\emph{resp.}, \emph{suppresses}) the  \emph{observable action} $\alpha$ of the trace $t\alpha$  of $J$ only if {\small $\confCPS{\Edit}{J} \trans{t\beta} N'$}, for some action $\beta$, with $\tau \neq \beta \neq \alpha$ (\emph{resp.}, with $\beta = \tau$). 
                If {\small $P \trans{t\alpha} Z$}, for \nolinebreak some \nolinebreak $Z$,
		then we say that there is a \emph{false positive} when trying the execution of $\alpha$. 
	\end{definition}

\paragraph{Behavioural equalities.}
In the paper, we adopt standard behavioural equivalences between (networks of) monitored controllers. In particular, we use \emph{trace equivalence}, written $\simeq$, \emph{weak similarity}, denoted $\sqsubseteq$, and \nolinebreak \emph{weak bisimilarity}, \nolinebreak   written \nolinebreak  $\approx$.

\section{Use case: a simple water transmission network}
\label{sec:case-study}

		\begin{figure*}[t]
	\centering
	\includegraphics[width=0.8\textwidth]{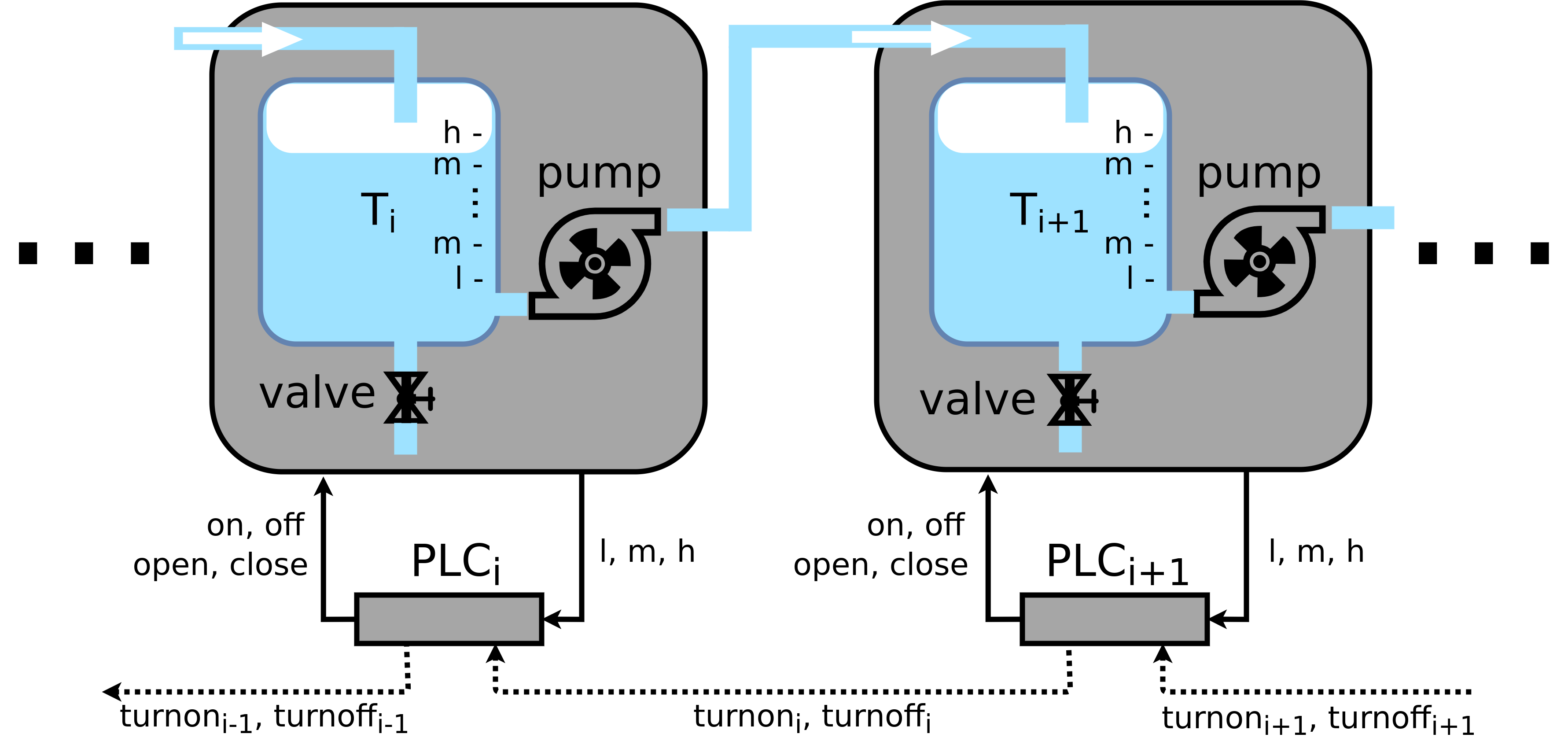}
	\caption{The typical structure of water transport networks}
	\label{figure:case-study-hankin}
\end{figure*}

  In this section,  we describe how to specify in our calculus \cname{} a non-trivial network of PLCs to control a water transmission network (WTN).

  Typical WTNs are composed of the following main physical elements: (i) tanks, (ii) pumping stations, (iii) water sources (\emph{e.g.}, boreholes), and (iv) pipes. In order, to monitor the status of each element, sensor devices are used to collect measurements regarding flow, pressure, level, and quality of the water that flows in the system. Figure~\ref{figure:case-study-hankin} gives us  a  typical configuration found in several water utilities, with the same structure replicated in larger infrastructures. 
  In this setup, borrowed from \cite{hankin_2020}, water is extracted from a water source (e.g., a borehole or another tank) using a pump. The pumps increase the water pressure which pushes the water into others tanks, which may be located a few kilometers away at a higher elevation. Each tank is equipped with a valve for the elimination of exceeding water.

  The finite-state machine control logic for WTNs is quite simple.
	In Table~\ref{table:plc-i-code}, we provide a possible code $P_i$ for the controller $\mathrm{PLC}_i$ managing the tank $T_i$ together with its physical devices (sensors and actuators). Here, the PLC waits for one time slot (to get stable sensor signals) and then checks the water level of the tank $T_i$, distinguishing between three possible states. 
	If the  water level is low (signal $l$) then the PLC sends a request of water to the sub-system $i{-}1$ via a channel transmission $\scriptstyle\overline{\req{on_{i-1}}}$, addressed to $\mathrm{PLC}_{i-1}$, requiring to turn on its pump. If the request is accepted then  $\mathrm{PLC}_i$ listens at both channels $\req{on_i}$ and $\req{off_i}$ for water requests coming from $\mathrm{PLC}_{i+1}$. Depending on whether these requests are accepted or not,  $\mathrm{PLC}_i$ will turn on/off  its pump (via commands $\overline{\on{ }}$ or $\overline{\off{ }}$, respectively), close the valve (via command $\overline\close$), and it will end its scan cycle. If there are no incoming requests from $\mathrm{PLC}_{i+1}$ then  $\mathrm{PLC}_{i}$ times out, closes the valve, and then ends the scan cycle. Similarly, if the water request $\scriptstyle\overline{\req{on_{i-1}}}$ is not accepted  by  $\mathrm{PLC}_{i-1}$ in the current time slot then it times out, closes the valve, and finally  ends its scan cycle.
\begin{table*}[t]
	
	\begin{displaymath}
	{\small
	\begin{array}{l}
	{P}_i	   \defn \fix \Xrec\big( \tick.   
	
	\timeout{l.\timeout{\overline{\req{on_{i-1}}}.\timeout{\req{on_{i}}.\overline{\on{} }.\overline{\close}.\fineC.{\Xrec}+ \req{off_{i}}.\overline{\off{}}.\overline{\close}.\fineC.{\Xrec}}{(\overline\close.\fineC.{\Xrec})}}{(\overline\close.\fineC.{\Xrec})}
		\\[1pt]
		\hspace*{10mm}
		+ \,
                	h.\timeout{\overline{\req{off_{i-1}}}.\timeout{\req{on_{i}}.\overline{\on{} }.\overline{\open}.\fineC.{\Xrec}+ \req{off_{i}}.\overline{\off{}}.\overline{\open}.\fineC.{\Xrec}}{(\overline\open.\fineC.{\Xrec})}}{(\overline\open.\fineC.{\Xrec})}}{(\fineC.{\Xrec})} \big )
		\\[1pt]
		\hspace*{10mm}			
		+ \, 
		m.\timeout{\req{on_{i}}.\overline{\on{} }.\fineC.{\Xrec}+ \req{off_{i}}.\overline{\off{}}.\fineC.{\Xrec}}{(\fineC.{\Xrec})}
	
	\end{array}
}
	\end{displaymath}
	\caption{The code of $\mathrm{PLC_i}$}
	\label{table:plc-i-code}
\end{table*}

If the water level of the tank  $T_i$ is high (signal $h$) then the behaviour of $\mathrm{PLC}_i$ is  specular to the previous case when a low level is detected (signal $l$).


	Finally, if the water of tank $T_i$ is at some intermediate level  between $l$ and $h$ (signal $m$) then $\mathrm{PLC}_i$ listens for water requests originating from $\mathrm{PLC}_{i+1}$ to turn on/off the pump. If it gets one of those requests in the current time slot then it reacts accordingly, otherwise it times out and ends the scan cycle.  More precisely, if $\mathrm{PCL}_i$ gets a $\req{on_i}$ request then it turns on the pump,  letting the water  flow from $T_i$ to $T_{i+1}$; otherwise, if it gets a $\req{off_i}$ request then 
	it turns off the pump; in both cases it ends the scan cycle and then returns.

Now, as our calculus can be also used to describe malicious code, in the following we provide a malware $M_i$ for  $\mathrm{PLC}_i$,  whose target is to empty the water tank $T_i$.  	
\begin{displaymath}
{ 
	\begin{array}{l}
	{M}_i	   =\fix \Xrec\big( \tick.\timeout{\overline{\req{off_{i-1}}}\timeout{\close.\Xrec}{\Xrec}
	}{\Xrec}\big)
	\end{array}
}
\end{displaymath}
The malware $M_i$ has a cyclic behaviour: it waits for one time slot and then sends a request at channel $\overline{\req{off_{i-1}}}$ to turn off the pump of the system $i{-}1$, pumping water  from $T_{i-1}$ to $T_i$; if the request is accepted then it drop the commands sent by $\mathrm{PLC}_i$  to close the valve, to completely empty the tank $T_i$. In this manner, the tank  $T_i$ will not
receive incoming water from the contiguous system $i{-}1$ and, at the same time, the closure of the valve of  $T_i$ is prevented; as consequence, the tank will finally get empty. 

\section{Monitor synthesis}
\label{sec:safety}


In Table~\ref{tab:synthesis}, we provide a synthesis function $\funEdit{-}$ that given a deterministic controller $P \in \mathbbm{Ctrl}$ returns a \emph{syntactically deterministic} edit automaton $\Edit \in \mathbbm{Edit}$ enforcing the correct behaviour of $P$, independently of the presence of an arbitrary malware $M \in  \mathbbm{Malw}$  that attempts to  \emph{inject} and/or \emph{drop} both \emph{actuator commands} and  \emph{channel communications}. 

 \begin{table}[t]
   \caption{The synthesis algorithm $\funEdit{-}$}
  {\small
 \begin{math}   
\begin{array}{lcl}
\funEdit{ \fix \Xrec S } & \defn & \fix \Xrec \funEdit{S} \\[5pt]

 \funEdit{\timeout{\sum_{i \in I}s_i.S_i}{S}} & \defn & \fix \Yrec \big(\sum_{i \in I}\eact{s_i}{s_i}.\funEdit{S_i}  + \eact{\tick}{\tick}.\funEdit{S}
+ \sum\limits_{\subalign{\scriptscriptstyle \alpha \in \ActSet^{\ast} \cup \ChanSet^{\ast}}}\eact{\alpha}{\tau}.\Yrec
\big) \\[5pt]

 \funEdit{\tick.S} & \defn & \fix \Yrec\big(\eact{\tick}{\tick}.\funEdit{S} 
+ \sum\limits_{\subalign{\scriptscriptstyle\alpha \in \ActSet^{\ast} \cup \ChanSet^{\ast}}}\eact{\alpha}{\tau}.\Yrec
\big) \\[5pt]

 \funEdit{\timeout{\sum_{i \in I}c_i.C_i}{C}} & 
\defn & \fix \Yrec\big(\sum_{i \in I} \eact{c_i}{c_i}.\funEdit{C_i} 
+ \eact{\tick}{\tick}.\funEdit{C}
+ \!\!\sum\limits_{\subalign{\scriptscriptstyle \alpha \in \ActSet^{\ast}}}\! \eact{\alpha}{\tau}.\Yrec 
+ \!\!\sum\limits_{\subalign{{\scriptscriptstyle \gamma\in \ChanSet^{\ast}  \setminus  \cup_{i \in I} \{ c_i \} }}} \! \eact{\gamma}{\tau}.\Yrec 
\big) \\[5pt]

 \funEdit{\timeout{\overline{c}.C_1}{C_2}} & \defn & \fix \Yrec\big(
\eact{\overline{c}}{\overline{c}}.\funEdit{C_1} + \eact{\tick}{\tick}.\funEdit{C_2} 
+ \sum\limits_{\subalign{\scriptscriptstyle \alpha \in \ActSet^{\ast}}}\eact{\alpha}{\tau}.\Yrec 
+ \sum\limits_{\subalign{\scriptscriptstyle \gamma\in \ChanSet^{\ast} \setminus \{ \overline{c} \}}} \eact{\gamma}{\tau}.\Yrec\big) \\[5pt]

 \funEdit{\overline{a}.A} & \defn & \fix \Yrec\big(\eact{\overline{a}}{\overline{a}}.\funEdit{A} + \eact{\tau}{\tau}.\Yrec 
+ \sum\limits_{\subalign{\scriptscriptstyle \alpha \in \ActSet^{\ast} \setminus\{ a, \overline{a}\}}}\eact{\alpha}{\tau}.\Yrec
+ \sum\limits_{\subalign{\scriptscriptstyle \gamma \in \ChanSet^{\ast}}}\eact{\gamma}{\tau}.\Yrec
\big)
\\[5pt]

 \funEdit{\fineC.\Xrec} & \defn & \fix \Yrec\big(\eact{\fineC}{\fineC}.\Xrec + \sum\limits_{\subalign{\scriptscriptstyle \alpha \in \ActSet^{\ast} \cup \ChanSet^{\ast}}}\eact{\alpha}{\tau}.\Yrec 
\big)
\end{array}
\end{math}
  }
    \label{tab:synthesis}
\end{table}

In the definition of our synthesis, we adopt the following standard notation for co-actions regarding actuator commands and channel communications: $\overline{\ActSet} \defn \{ \overline{a} \mid a \in \ActSet \}$ and  
$\overline{\ChanSet} \defn \{ \overline{c} \mid c \in \ChanSet \}$. Furthermore, we define $\ActSet^{\ast} \defn \ActSet \cup \overline{\ActSet}$ and $\ChanSet^{\ast} \defn \ChanSet \cup \overline{\ChanSet}$. 

 Let us comment on the details of the synthesis function $\funEdit{-}$ of  Table~\ref{tab:synthesis}. The edit automaton associated to  listening on sensor signals allows all incoming signals expected by the controller, together with the passage of time due to eventual timeouts. All other actions are suppressed. The edit automaton associated to the listening on communication channels is similar, except that  communications that are not admitted by the controller are suppressed to prevent both \emph{drops and injections on system channels}, as well as, covert communications between \emph{colluding malware} running in different PLCs. Channel transmissions are allowed only when occurring, in the right order, on those channels intended by the controller; all other actions are suppressed. Only genuine actuator commands (again, in the right order) are allowed. \emph{Drops of actuator commands}, the only possible intra-controller interaction occurring between the genuine controller and the malware, are allowed because we want an observation-based monitoring.  Finally, the monitoring edit automaton \nolinebreak and the associated controller do synchronise at the end of each controller cycle via the \nolinebreak action $\fineC$: all other actions emitted by the compromised controller are suppressed, included those actions coming from the genuine controller that was left behind in its execution due to some injection attack  mimicking (part of) some correct behaviour. We recall that only the construct $\fineC.\Xrec$ (and not $\fineC.P$) is  used to write PLC programs.

 As an example, in Table~\ref{tab:example} we provide the edit automaton resulting from our synthesis algorithm applied to the PLC introduced in our use case. For simplicity, with a small abuse of notation, we used parametric processes. 
 \begin{table}[t]
   \caption{Edit automaton synthesised from the code $P_i$ of PLC$_i$ of Section~\ref{sec:case-study}}
   \begin{displaymath}
{\small 
	\begin{array}{rcl}
	\funEdit{
		P_i} & \defn & \fix \Xrec  \fix \Yrec \big(\eact{\tick}{\tick}. \mathrm{ChkLvl}
	+ \sum\limits_{\subalign{\alpha \in \ActSet^{\ast} \cup \ChanSet^{\ast}  }}\eact{\alpha}{\tau}.\Yrec  
	\big) \\
	
	\mathrm{ChkLvl}  &
	\defn & \fix \Yrec \big( \eact{l}{l}.\mathrm{Req}_{l} + 
 \eact{h}{h}.\mathrm{Req}_{h} 
 +
 	\eact{m}{m}. \mathrm{Req}_{m} 
	+ \eact{\tick}{\tick}.\mathrm{End} 
	+ \!\!\!\sum\limits_{\subalign{\alpha \in \ActSet^{\ast} \cup \ChanSet^{\ast} }}\eact{\alpha}{\tau}.\Yrec  
	\big) \\
	
	\mathrm{Req}_{l} & 
	\defn & \fix \Yrec\big( \eact{\overline{\req{on_{i-1}}}}{\overline{\req{on_{i-1}}}}.\mathrm{C}\langle {{\close }}\rangle 
	+ \eact{\tick}{\tick}. \mathrm{A}\langle{{\close}}\rangle
	+ \sum\limits_{\subalign{\alpha \in  \ActSet^{\ast} \cup \ChanSet^{\ast} \setminus\{\scriptscriptstyle \overline{\req{on_{i-1}}}\}}}\eact{\alpha}{\tau}.\Yrec  
	\big) \\

	\mathrm{Req}_{h} & 
	\defn & \fix \Yrec\big( \eact{\overline{\req{off_{i-1}}}}{\overline{\req{off_{i-1}}}}.\mathrm{C}\langle {{\open }}\rangle 
	+ \eact{\tick}{\tick}. \mathrm{A}\langle{{\open}}\rangle
	+ \sum\limits_{\subalign{\alpha \in  \ActSet^{\ast} \cup \ChanSet^{\ast} \setminus\{\scriptscriptstyle \overline{\req{off_{i-1}}}\}}}\eact{\alpha}{\tau}.\Yrec  
	\big) \\

        	\mathrm{Req}_{m} & 
	\defn & \fix \Yrec\big( \eact{\req{on_{i}}}{\req{on_{i}}}.\mathrm{A}\langle {\on{}} \rangle
	+
	\eact{\req{off_{i}}}{\req{off_{i}}}.\mathrm{A}\langle {\off{}} \rangle 
	+ \eact{\tick}{\tick}.\mathrm{End}
	+ \sum\limits_{\subalign{\alpha \in  \ActSet^{\ast} \cup \ChanSet^{\ast} \setminus\{\scriptscriptstyle \req{on_{i}},\req{off_{i}}\}}}\eact{\alpha}{\tau}.\Yrec  
	\big) \\

	\mathrm{A}(a)  & \defn & \fix \Yrec  \big(\eact{\overline{a}}{\overline{a}}. \mathrm{End}
	+ \eact{\tau}{\tau}.\Yrec
	+ \sum\limits_{\subalign{\scriptscriptstyle \alpha \in \ActSet^{\ast} \setminus\{ a, \overline{a}\}}}\eact{\alpha}{\tau}.\Yrec
	+ \sum\limits_{\subalign{\scriptscriptstyle \gamma \in \ChanSet^{\ast}}}\eact{\gamma}{\tau}.\Yrec  
	\big) \\

        	\mathrm{C}(a) & 
	\defn & \fix \Yrec\big( \eact{\req{on_{i}}}{\req{on_{i}}}.\mathrm{On}\langle{a}\rangle
	+
	\eact{\req{off_{i}}}{\req{off_{i}}}.\mathrm{Off}\langle{a}\rangle 
	+ \eact{\tick}{\tick}.\mathrm{A}\langle{a}\rangle
	+ \sum\limits_{\subalign{\alpha \in  \ActSet^{\ast} \cup \ChanSet^{\ast} \setminus\{\scriptscriptstyle \req{on_{i}},\req{off_{i}}\}}}\eact{\alpha}{\tau}.\Yrec  
	\big) \\

	\mathrm{On}(a) & \defn & \fix \Yrec  \big(\eact{\overline{\on{}}}{\overline{\on{}}}. \mathrm{A}\langle {a}\rangle
	+ \eact{\tau}{\tau}.\Yrec
	+ \sum\limits_{\subalign{\scriptscriptstyle \alpha \in \ActSet^{\ast} \setminus\{ \on{}, \overline{\on{}}\}}}\eact{\alpha}{\tau}.\Yrec
	+ \sum\limits_{\subalign{\scriptscriptstyle \gamma \in \ChanSet^{\ast}}}\eact{\gamma}{\tau}.\Yrec  
	\big) \\
	
	\mathrm{Off}(a)  & \defn & \fix \Yrec  \big(\eact{\overline{\off{}}}{\overline{\off{}}}. \mathrm{A}\langle {a}\rangle
	+ \eact{\tau}{\tau}.\Yrec
	+ \sum\limits_{\subalign{\scriptscriptstyle \alpha \in \ActSet^{\ast} \setminus\{ \off{}, \overline{\off{}}\}}}\eact{\alpha}{\tau}.\Yrec
	+ \sum\limits_{\subalign{\scriptscriptstyle \gamma \in \ChanSet^{\ast}}}\eact{\gamma}{\tau}.\Yrec  
	\big) \\

	\mathrm{End}  & \defn & \fix \Yrec\big(\eact{\fineC}{\fineC}.\Xrec 
	+ \sum\limits_{\subalign{\alpha \in \ActSet^{\ast} \cup  \ChanSet^{\ast} }}\eact{\alpha}{\tau}.\Yrec  
	\big) 
	\end{array}
}
   \end{displaymath}
   \label{tab:example}
 \end{table}

Before proving the technical properties of the enforcement mechanism induced by our synthesised monitors, we focus our attention on two easy observations. 
\begin{remark}[Observation-based monitoring]
Our monitoring  is  \emph{observation-based} as the edit automata resulting from our synthesis never  correct $\tau$-actions (\emph{i.e.}, non-observable actions). 
\end{remark}

\begin{remark}[Colluding malicious activities]
\label{rem:colluding}
Any  inter-controller activity which does not comply  with the genuine behaviour of the PLC under scrutiny is suppressed by the enforcement.
\end{remark}

The synthesis proposed in Table~\ref{tab:synthesis} is suitable for implementation. 
\begin{proposition}[Determinism preservation]
  \label{prop:deterministic}
  Let $P \in \mathbbm{Ctrl}$ be a deterministic controller.  The  automaton
  $\funEdit{P}$ is syntactically deterministic in the  \nolinebreak sense \nolinebreak of~\cite{Acetoetal2017}.
\end{proposition}
Furthermore,  our synthesis algorithm is  computationally feasible. The complexity of the synthesis  is quadratic on the dimension of the controller, 
 where, intuitively, the dimension of a controller $ P \in \mathbbm{Ctrl}$,  
written $\ctrlDim{P}$, is  given by the number of prefixes $\alpha \in \overline{\ActSet}\cup\ChanSet^{\ast}\cup\SensSet \cup \{\tick, \fineC\}$ occurring in it (its formal definition can be found in the appendix).

\begin{proposition}[Polynomial complexity]
  \label{prop:poly}
  Let $P \in \mathbbm{Ctrl}$ be a deterministic controller, the complexity
  to synthesise $\funEdit{P}$ is $\mathcal{O}(n^2)$, with $n=\ctrlDim{P}$.  
\end{proposition}

As required at the beginning of this section, the synthesised edit automata are always \emph{transparent}, \emph{i.e.}, they never introduce non-genuine behaviours. 
\begin{proposition}[Transparency]
\label{prop:transparency}
If $P \in \mathbbm{Ctrl} $ then 
\(
\confCPS{\funEdit{P}}{P} \, \approx \, \confCPS{\go}{P}  . 
\)
\end{proposition}

%

Furthermore, our enforcement enjoys  
\emph{soundness preservation}: in a monitored controller, a malware may never trigger an incorrect behaviour.

\begin{proposition}[Soundness]
\label{prop:safety}
\hspace*{-2mm} Let $P$ be an arbitrary controller and  $M$ be  an arbitrary malware. Then, 
$\confCPS{\funEdit{P}}{ P | M} \: \sqsubseteq \: \confCPS{\funEdit{P}}{P}$. 
\end{proposition}
%

In the next proposition, we provide a result that is somehow complementary to Proposition~\ref{prop:safety}. The intuition being that in a monitored controller $ \confCPS{\funEdit{P}}{ P | M}$ the controller $P$ \emph{may}  execute all its (genuine) execution traces even in the presence of an arbitrary malware $M$.  Said in other words, the controller $P$ has always a chance to follow (and complete) its correct execution, even when compromised by the presence of a malware $M$.

\begin{proposition}
\label{prop:completeness}
Let $P$ be an arbitrary controller and $M$ be an arbitrary malware. 
Then, 
$
  \confCPS{\funEdit{P}}{ P | M} \: \sqsupseteq \: \confCPS{\funEdit{P}}{P}$. 
\end{proposition}
%


By applications of Propositions~\ref{prop:transparency}, \ref{prop:safety}, 
and \ref{prop:completeness} we can summarise our enforcement in a single equation.
\begin{theorem}[Trace enforcement]
\label{thm:w-enf-ctrl}
Let $P \in \mathbbm{Ctrl}$ be an arbitrary controller and  $M \in \mathbbm{Malw}$ be an arbitrary malware. 
Then, 
\(
 \confCPS{\funEdit{P}}{P | M} \: \simeq \: \confCPS{\go}{P} \, . 
\)
\end{theorem}

	An easy consequence of trace enforcement is the absence of \emph{false negatives}. 
	\begin{proposition}[Anomaly detection]
			\label{cor:false-negatives}
			Let $P$ be an arbitrary  controller,  $M$ be  an arbitrary malware, $J=P |M$, and $t$ be a genuine trace of $P$ (\emph{i.e.},  $P \trans{t}Z$, for some $Z$). If    $J \trans{t  \alpha}J' $, for some  $\alpha$ and $J'$, but $P \trans{ t \alpha} Z$ for no $Z$,   then  
			$\funEdit{P}$ promptly detects an anomaly of $J$  when  trying the execution of   $\alpha$ in  the trace \nolinebreak 
			$t  \alpha$. 
		\end{proposition}

Furthermore, trace enforcement scales to \emph{field communications networks} of communicating controllers compromised by  the presence of \emph{colluding malware}. 
\begin{proposition}[Trace enforcement of  field networks]
\label{cor:weak-enforcement}
Let $P_i \in \mathbbm{Ctrl}$  and  $M_i \in \mathbbm{Malw}$, for $1 \leq i \leq n$. 
Then, 
\begin{math}
\prod_{i=1}^n \confCPS{\funEdit{P_i}}{P_i  | M_i} \; \simeq \; \prod_{i=1}^n \confCPS{\go}{P_i}   \, . 
\end{math}
\end{proposition}
%

However, trace enforcement has a non-trivial inconvenient: it does not enjoy deadlock-freedom in the presence of a malware injecting correct actions. Let us formalise such a situation in the next remark.  
%
\begin{remark}[Injection attacks may prevent deadlock-freedom]
\label{rem:liveness}
In a monitored controller of the form $\confCPS{\funEdit{P}}{P | M}$, it may well  happen that the malware $M$  misleads the edit automaton $\funEdit{P}$ by injecting a trace $M \trans{\alpha_1} \ldots \trans{\alpha_n} M'$ of actions, with $\alpha_i \neq \tick$, compatible with the correct behaviour of the controller; in the sense that the very same trace may be executed by $P$: $P \trans{\alpha_1} \ldots \trans{\alpha_n} Q$, for some state $Q$.  This would give rise to the following admissible execution trace for the monitored controller: 
\( \confCPS{\funEdit{P}}{P |  M} \trans{\alpha_1} \ldots \trans{\alpha_n} \confCPS{\funEdit{Q}}{P | M'} \), 
in which the actual controller $P$ remains inactive. At that point, if the malware $M'$ suddenly stops mimicking an admissible behaviour of the controller,
the edit automaton $\funEdit{Q}$  will suppress all possible actions, even those proposed by  $P$,  which was left behind in its execution. Thus,  the monitored controller will continue its evolution as follows: 
\( \confCPS{\funEdit{Q}}{P | M'} \trans{\tau} \ldots \trans{\tau} \confCPS{\funEdit{Q}}{P' | M''} \). 
In this case, as neither the controller nor the malware can give rise to zeno behaviours, the enforced system may eventually reach a \emph{deadlock state} in which (i) $P'=\fineC.\Xrec$, (ii) $M''=\tick.M'''$, for  some $M'''$, or $M''=\nil$, and (iii) $\funEdit{Q}$ does not allow $\fineC$-actions because it requires some actions (\emph{e.g.}, actuations) to be performed before ending the scan cycle.
\end{remark}
%
Notice that  Remark~\ref{rem:liveness} is not in contradiction with Proposition~\ref{prop:completeness} because in that proposition we proved that a controller  has  a chance to follow and complete its correct behaviour in the presence of an arbitrary malware. Here,  we say a different thing:  a  malware  has a chance to deadlock  our monitored controllers. 
%

\section{Mitigation: the recipe for deadlock-freedom}
\label{sec:liveness}
In this section, we introduce an extra transition rule for monitored controllers to implement 
\emph{mitigation}, \emph{i.e.}, the \emph{insertion} of a sequence of activities  driven by \nolinebreak the \nolinebreak edit automaton in full autonomy, when the controller has lost contact with its enforcer: 
\[{
\Txiom{Mitigation}
{  J \trans{\fineC}J' \Q  \Edit \trans{\eact{\alpha}{\alpha}}  \Edit' \Q \alpha \in \ChanSet^{\ast} \cup \overline{\ActSet} \cup \{ \tick \}} 
{  \confCPS {\Edit}{  \mathit{J}} \trans{\alpha}  \confCPS {\Edit'}{\mathit{J} }}
}
\]
Intuitively, if the compromised controller signals the end of the scan cycle by emitting the action $\fineC$ 
and, at the same time, the current edit automaton $\Edit$ is not in the same state, 
then $\Edit$ will command the execution of a safe trace, without any involvement of the  controller, to reach the end of the controller cycle. When both the controller and the edit automaton will be aligned (at the end of the cycle)  they will synchronise on the action $\fineC$, via an application of the transition rule \rulename{Enforce}, and from then on they will continue in a tethered fashion.


Notice that in a monitored controller  $\confCPS{\Edit}{J}$ where $J$ is corrupted by some malware, the two components   $\Edit$ and $J$ may get  misaligned as they may reach unrelated states. For instance, in case of drop of actuator commands the corrupted controller $J$  may  reach an incorrect  state, leaving behind  its monitoring edit automata $\Edit$. 
In this case, the remaining observable actions in the current cycle of the compromised controller will be suppressed until the controller reaches the end of the scan cycle, signalled by the emission of an  $\fineC$-action 
(notice that since our malware are time-guarded they cannot introduce zeno behaviours to prevent a controller to reach the end of its scan cycle).
Once the compromised controller has been driven to \nolinebreak  the end of its cycle, the transition rule \rulename{Mitigation} goes into action.
%
      \begin{remark}
          \label{rem:mitigation-time} 
        The assumption made in Remark~\ref{rem:maximum-time} ensures us enough time to complete the mitigation of the scan cycle, well before  the maximum cycle limit. 
\end{remark}

As a main result, we prove that with the introduction of the rule  \rulename{Mitigation} our runtime enforcement for controllers works faithfully up to weak bisimilarity, ensuring  deadlock-freedom. 
\begin{theorem}[Observational enforcement]
	\label{thm:s-enf-ctrl}
	Let $P \in \mathbbm{Ctrl}$ be an arbitrary controller and  $M \in \mathbbm{Malw}$ be an arbitrary malware. 
	Then, 
\(
\confCPS{\funEdit{P}}{P | M} \, \approx \, \confCPS{\go}{P} \, . 
\)
\end{theorem}
%

	A consequence of Theorem~\ref{thm:s-enf-ctrl} is  the prompt detection and mitigation of alterations of PLC behaviours in the presence of injected malware.

	\begin{proposition}[Anomaly detection and mitigation]
		\label{cor:false-negatives-2}
		Let $P \in \mathbbm{Ctrl}$ be an arbitrary controller,  $M \in \mathbbm{Malw}$ be  an arbitrary malware, and $J=P|M$ the derived compromised controller. 
		\begin{enumerate}
			\item  
			If    {\small $J \trans{t  \alpha}J'$} for some genuine trace $t$ of $P$ (\emph{i.e.}, {\small $P \trans{t} Z$}, for some $Z$), 
			for some \emph{observable action} $\alpha$,  but {\small $P \trans{ t \alpha} Z$} for no $Z$,   then the monitor $\transf{P}$ \emph{detects} an anomaly of $J$  when trying the execution of  the incorrect  action $\alpha$ of the trace $t  \alpha$. 
			\item 
			Whenever $\transf{P}$  detects an anomaly $\alpha$ \nolinebreak in  $J$, it \emph{mitigates} the anomaly either by 
			correcting the action \nolinebreak $\alpha$  with an action $\beta$, $\tau \neq \beta \neq\alpha$, or by suppressing \nolinebreak the \nolinebreak action \nolinebreak $\alpha$.
		\end{enumerate}
	\end{proposition}

As for trace enforcement, observational enforcement  scales to \emph{field networks} of communicating controllers compromised by  the presence of (potentially) \emph{colluding malware}. 
\begin{corollary}[Observational enforcement of  field networks]
  \label{cor:obs-enforcement}
Let $P_i \in \mathbbm{Ctrl}$  and  $M_i \in \mathbbm{Malw}$, for $1 \leq i \leq n$. 
Then, 
\begin{math} 
\prod_{i=1}^n \confCPS{\funEdit{P_i}}{P_i | M_i} \; \approx \; \prod_{i=1}^n \confCPS{\go}{P_i}   \, . 
\end{math}
\end{corollary}

Now, we show an easy application of Corollary~\ref{cor:obs-enforcement} to the (simplified) water transmission network seen in Section~\ref{sec:case-study}. We recall the synthesis in Table~\ref{tab:example} obtained from (the code of) $\mathrm{PLC_i}$.
\begin{proposition}
For any arbitrary malware $M_i$, for $1 \leq i \leq n$, 
\begin{displaymath}
{\small  \prod_{i=1}^{n}\confCPS{\funEdit{\mathrm{PLC}_i}}{\mathrm{PLC}_i | M_i } \, \approx \, \prod_{i=1}^{n}\confCPS{\go}{\mathrm{PLC}_i} \, . 
}
\end{displaymath}
\end{proposition}
In particular, the proposition above holds  for the example of malware code proposed at the end of Section~\ref{sec:case-study}.

\section{Conclusions and related work}
\label{sec:conclusion}
We have defined a formal language to express  networks of monitored PLCs, potentially compromised with colluding malware that may forge/drop actuator commands and inter-controller communications. We do not deal with alterations of sensor signals  within a PLC, as they can already be  altered either at the network level or within the sensor devices~\cite{ACM-survey2018}.

The runtime enforcement has been achieved via a finite-state sub-class of Ligatti's edit automata equipped with an ad-hoc operational semantics to deal with \emph{system mitigation},  by inserting actions in full autonomy when the monitored controller is not able to do so in a correct manner.

Then, we have provided a synthesis algorithm that, given a deterministic uncorrupted controller, returns, in polynomial time, a syntactically deterministic  edit automata  to enforce the correctness of the controller. The proposed enforcement meets a number of requirements: observation-based monitoring,  transparency, soundness,  deadlock-freedom,  and both detection and mitigation of alterations of the behaviour of the monitored PLC in case of injected malware.

\paragraph*{Related work.}
	The notion of \emph{runtime enforcement} was introduced by Schneider~\cite{Schneider2000} to enforce security policies. These properties are enforced by means of  \emph{truncation automata}, a kind of automata that terminates the monitored system in case of violation of the property.
Thus, truncation automata can only enforce safety properties.
 Ligatti et al.~\cite{Ligatti2005} extended Schneider's work by proposing the notion of  \emph{edit automaton}, \emph{i.e.}, an enforcement mechanism able of \emph{replacing}, \emph{suppressing}, or even \emph{inserting} system actions. Edit automata are capable of enforcing instances of safety and liveness properties, along with other properties such as renewal properties~\cite{Bielova2011,Ligatti2005}. 
In general, Ligatti et al.'s edit automata have an enumerable number of states,  whereas in the current paper we restrict ourselves to finite-state edit automata. Furthermore, in its original definition the insertion of actions is possible at any moment, whereas our monitoring edit automata can insert actions, via the rule \rulename{Mitigation}, only when the PLC under scrutiny reaches a specific state, \emph{i.e.}, the end of the scan cycle. Notice that our actions of the form $\eact{\alpha}{\beta}$ can be easily expressed in the original formulation by inserting the action $\beta$ and then suppressing the action $\alpha$. 
Unlike Schneider and Ligatti et al., we do not enforce specific properties for all admissible systems (in our case, controllers) but we ensure the preservation of the correct semantics of a corrupted controller. 
Bielova and Massacci~\cite{Bielova2011,Bielova2011Jrn} 
 provided a stronger notion of enforceability by introducing a \emph{predictability} criterion to prevent monitors from transforming invalid executions in an arbitrary manner. Intuitively, a monitor is said predictable if one can predict the number of transformations used to correct invalid executions, thereby avoiding unnecessary transformations. In our case, we never introduce unnecessary transformations as our synthesis is based on the exact knowledge of the controller. 
	Falcone et al. \cite{Falcone2,Falcone3} proposed a synthesis algorithm, relying on \emph{Streett automata}, to translate most of the property classes defined within the \emph{safety-progress hierarchy}~\cite{Hierarchy_of_Temporal}
        into enforcers. 
K\"onighofer et al.~\cite{Konighofer2017} proposed a synthesis algorithm that given a safety property returns a monitor, called \emph{shield},  that analyses both inputs and outputs of reactive systems in order to enforce properties by modifying the outputs only.   Pinisetty et al.~\cite{Pinisetty_2017} have proposed a bi-directional runtime enforcement mechanism  for reactive systems, and more generally for cyber-physical systems, to correct both inputs and outputs. 
Aceto et al.~\cite{AcetoCC2018}  developed an operational framework to  enforce properties in HML logic with recursion ($\mu$HML) relying on suppression. More precisely, they achieved the  enforcement of  a  safety  fragment of $\mu$HML by providing a linear automated synthesis algorithm that generates correct suppression monitors from formulas.  Enforceability of  modal $\mu$-calculus (a reformulation of $\mu$HML) was previously tackled  by Martinelli and Matteucci~\cite{MartinelliMatteucci2007} by means of a synthesis algorithm which is exponential in the lenght of the enforceable formula.   More recently,  Cassar~\cite{Cassar_PhD} defined a general framework to compare different enforcement models and different correctness criteria, including optimality. His works focuses on the enforcement of a safety fragment of $\mu$HML, paying attention to both directional and bi-directional notions of enforcement. 

The present work is a revised extension of the extended abstract appeared in~\cite{ICTCS2020}. Here, besides full proofs, we provide new results on the anomaly detection and the mitigation activity of the monitoring secured proxy. In particular, as formally stated in Proposition~\ref{cor:false-negatives} and Proposition~\ref{cor:false-negatives-2}, our secured proxies promptly detect abnormal behaviours and safely intervene by mitigating them. Furthermore, compared to the conference paper, in Section~\ref{sec:case-study} we provide a more involved and realistic use case, taken from the field of water transmission networks. We then apply our synthesis algorithm to derive an enforcing edit automata for the family of PLCs proposed in our use case. Finally, in our companion paper~\cite{CSF2020} we abstracted over PLC implementations and provided a simple language of regular properties to express correctness properties that should be enforced upon completion of PLC scan cycles.

\bibliography{main}
\bibliographystyle{splncs04}

\appendix
\section{Proofs}

Before proving Proposition~\ref{prop:poly} we provide a formal definition of the size of a controller.
\begin{definition}
	For a generic controller $Z \in \mathbbm{Ctrl} \cup \mathbbm{Sens} \cup \mathbbm{Comm} \cup \mathbbm{Act} $, we define the size of $Z$, written  $\ctrlDim {Z}$,  by induction on the structure of the controller:
	{\small
		\begin{displaymath}
		\begin{array}{lcl@{\hspace*{1cm}}lcl}		
		\ctrlDim{ \fix \Xrec S } &\defn& \ctrlDim{ S }
		&\ctrlDim{\timeout{\sum_{i \in I}s_i.S_i}{S}} &\defn& | I | + \sum_{i \in I}\ctrlDim{S_i} + \ctrlDim{S}
		\\[2pt]
		
		\ctrlDim{\tick.S} &\defn &  1+\ctrlDim{S}  &\ctrlDim{\timeout{\sum_{i \in I}c_i.C_i}{C}} 
		&\defn&  | I | + \sum_{i \in I}\ctrlDim{C_i} + \ctrlDim{C}
		\\[2pt]

		\ctrlDim{\overline{a}.A} &\defn&  1 + \ctrlDim{A} &
		\ctrlDim{\timeout{\overline{c}.C_1}{C_2}} &\defn&  1 +\ctrlDim{C_1} + \ctrlDim{C_2} \\[2pt]
		\ctrlDim{\fineC.\Xrec} &\defn&    1 \, .
		\end{array}
		\end{displaymath}
	}	
\end{definition}

 \emph{Let us prove Proposition~\ref{prop:poly}.}

\begin{proof}
	For a generic controller $Z$, we prove that the recursive structure of the synthesis function $\funEdit{Z'}$ with $Z' \in \{P, S, C,A\}$ can be characterised by the following form: $T(m) = T(m-1) + n$ , with $n = \ctrlDim{Z}$ and $\ctrlDim{Z'}=m\leq n$ . Hence the thesis follows since $T(n) = T(n-1) + n$   is $\mathcal{O}(n^2)$.
	We prove this by case analysis on the structure of $Z$ by examining each synthesis step in which the synthesis function is processing $m = \ctrlDim{Z'}$ symbols, whit $m\leq n$ and $Z' \in \{P, S, C,A\}$. Thus, we characterise: \emph{(i)} how many symbols of $Z$ the synthesis functions processes, \emph{(ii)} how many times the synthesis function calls itself and \emph{(iii)} how many computations performs in that step. 
	We consider the most significant cases.
	

{\bf Case} $\timeout{\sum_{i \in I}c_i.C_i}{C}$. 
 For $m=\ctrlDim{\timeout{\sum_{i \in I}c_i.C_i}{C}}$, by definition the synthesis function consumes all $c_i$, with $i\in I$. The synthesis function calls itself again $r+1$ times where $r = |I|$. 
 Thus, each $\funEdit{C_i}$ operates on $\ctrlDim{C_i}$ remaining symbols and $\funEdit{C}$ operates on $\ctrlDim{C}$ remaining symbols. The synthesis function produces a sum over $\alpha \in \ActSet^{\ast} \cup \ChanSet^{\ast} \setminus \bigcup_{i \in I}c_i$ which are at most $n$ operations. 
 Thus, we can characterise the recursive structure as $T(m) =  \sum_{i \in I} T(  \ctrlDim{C_i}) + T(\ctrlDim{C})  + n $. Since $\sum_{i \in I} \ctrlDim{C_i} + \ctrlDim{C} = m-{|I|-1}\le m-1$, the complexity  is smaller than that  of $T(m-1)+n$.


{\bf Case } $\overline{a}.A$.  
For $m=\ctrlDim{\overline{a}.A}$, by definition the synthesis function consumes the, $\overline{a}$ and calls itself again once $\funEdit{A}$ Moreover, the synthesis function produces a sum over $\alpha \in \ActSet^{\ast} \cup \ChanSet^{\ast}  \cup \{\tau\} \setminus\{ a, \overline{a}\}$ which are at most $n$ operations. Thus we can characterise the recursive structure as: $T(m) = T(m-1) + n$. 

\end{proof}

\emph{Let us prove Proposition~\ref{prop:safety} (Soundness)}.
\begin{proof}
	Let us define four  binary relations:
	\begin{itemize}
		\item ${\mathcal P} \defn \{ ( \confCPS{\funEdit{P}}{J}, \confCPS{\funEdit{P}}{P}) \, \mid  \textrm{ for any $P$ and $J$} \} $; 
		\item ${\mathcal S} \defn \{ ( \confCPS{\funEdit{S}}{J}, \confCPS{\funEdit{S}}{S}) \, \mid  \textrm{ for any $S$ and $J$} \} $; 
		\item ${\mathcal C} \defn \{ (\confCPS{\funEdit{C}}{J} ,  \confCPS{\funEdit{C}}{C}) \, \mid  \textrm{ for any $C$ and $J$} \} $; 
		\item ${\mathcal A} \defn \{ ( \confCPS{\funEdit{A}}{J} ,  \confCPS{\funEdit{A}}{A} ) \, \mid  \textrm{ for any $A$ and $J$} \} $. 
	\end{itemize}
	We prove that the relation $\rel \defn {\mathcal P} \cup {\mathcal S} \cup {\mathcal C} \cup {\mathcal A}$ is a \emph{weak simulation}. For each pair $( N_1,  N_2 ) \in \rel$ we proceed by case analysis on why $N_1 \trans{\alpha} N_1'$. 
	We consider the most significant cases.

	\noindent
	{\bf Case} $(N_1,N_2) \in {\mathcal C}$.  We do case analysis on the structure of the controller $C$ in $N_1= \confCPS{\funEdit{C}}{J}$, for some arbitrary $J$.  
	
\noindent
 Let $C \equiv \timeout{\sum_{i} c_i.C_i}{\hat{C}}$. 
		\begin{itemize}
			\item Let $N_1 \trans{c_i} \confCPS{\funEdit{C_i}}{J'} = N'_1$, by an application rule \rulename{Enforce} as $J \trans{c_i} J'$. Then, $N_2=\confCPS{\funEdit{C}}{C} \trans{c_i} \confCPS{\funEdit{C_i}}{C_i} = N_2'$ and $(N_1',N_2') \in {\mathcal C} \subset \rel$ by construction. 
			
			\item Let $N_1 \trans{\tick} \confCPS{\funEdit{\hat{C}}}{J'} = N'_1$, by an application rule \rulename{Enforce} as $J \trans{\tick} J'$. Then, $N_2 = \confCPS{\funEdit{C}}{C} \trans{\tick} \confCPS{\funEdit{\hat{C}}}{\hat{C}} = N'_2$ and $(N_1',N_2') \in {\mathcal C} \subset \rel$ by construction.
			
			\item  Let $N_1 \trans{\tau} \confCPS{\funEdit{C}}{J'} = N_1'$, by an application rule \rulename{Enforce}, because $J \trans{\alpha} J'$, for some $\alpha \in \ActSet^{\ast}\cup \ChanSet^{\ast} \setminus  \bigcup_i c_i$.   Here, notice that the edit automaton $\funEdit{C}$ \emph{suppresses} all  possible injections originating from the malware, turning them into $\tau$-actions. 
			Thus, $N_2 = \confCPS{\funEdit{C}}{C} \Trans{\hat{\tau}} \confCPS{\funEdit{C}}{C} = N_2'$ and $(N_1',N_2') \in {\mathcal C} \subset \rel$ by construction.
		\end{itemize}
		
\noindent
Let $C \equiv \timeout{\overline{c}.C_1}{C_2}$.  This case is similar to the previous one.
	 
\noindent  Let $C \equiv A$. In this case, we end up to the case  $(N_1,N_2) \in {\mathcal A}$.

	\noindent
	{\bf Case} $(N_1,N_2) \in {\mathcal A}$.  We do case analysis on the structure of the controller $A$ in $N_1= \confCPS{\funEdit{A}}{J}$, for some arbitrary $J$.   

\noindent  Let $A \equiv \overline{a}.\hat{A}$.  
		\begin{itemize}
			\item Let $N_1 \trans{\overline{a}} \confCPS{\funEdit{\hat{A}}}{J'} = N'_1$, by an application of \rulename{Enforce} as $J \trans{\overline{a}} J'$ Then, $N_2=\confCPS{\funEdit{A}}{A} \trans{\overline{a}} \confCPS{\funEdit{\hat{A}}}{\hat{A}} = N_2'$ and $(N_1',N_2') \in {\mathcal A} \subset \rel$ by construction. 
			
			\item  Let $N_1 \trans{\tau} \confCPS{\funEdit{\hat{A}}}{J'} = N_1'$, by an application rule \rulename{Enforce}, because $J \trans{\alpha} J'$, for some $\alpha \in \{\tau\}  \cup  \ActSet^{\ast}\cup \ChanSet^{\ast} \setminus  \{\overline{a}\}$. Note that the edit automaton $\funEdit{A}$ \emph{suppresses} all  possible injections originating from the malware, turning them into $\tau$-actions. 
			Thus, $N_2 = \confCPS{\funEdit{\hat{A}}}{\hat{A}} \Trans{\hat{\tau}} \confCPS{\funEdit{\hat{A}}}{\hat{A}} = N_2'$ and $(N_1',N_2') \in {\mathcal A} \subset \rel$ by construction.
		\end{itemize}
		
\noindent 
 Let $A \equiv \fineC.\Xrec$.

		\begin{itemize}
			\item Let $N_1 \trans{\fineC} \confCPS{\funEdit{S}}{J'} = N'$, by an application rule \rulename{Enforce} triggered by \rulename{Ctrl}. By definition of the transition rule \rulename{Rec} we end up from $\fineC.\Xrec$ to the case  $\fineC.S$. Then, $N_2=\confCPS{\funEdit{A}}{A} \trans{\fineC} \confCPS{\funEdit{S}}{S} = N_2'$ and $(N_1',N_2') \in {\mathcal S} \subset \rel$ by construction.
			
			\item Let $N_1 \trans{\tau} \confCPS{\funEdit{A}}{J'} = N_1'$, by an application rule \rulename{Enforce} triggered by an application of rule \rulename{Inject} or \rulename{Ctrl}, because $J \trans{\alpha} J'$ for some $\alpha \in \{a, \overline{a}, c, \overline{c}\}$.   Here, notice that the edit automaton $\funEdit{S}$ \emph{suppresses} all  possible injections originating from the malware or the controller not aligned with the edit, turning them into $\tau$-actions. 
			Thus, $N_2=\confCPS{\funEdit{A}}{A} \Trans{\hat{\tau}} \confCPS{\funEdit{A}}{A} = N_2'$ and $(N_1,N_2') \in {\mathcal A} \subset \rel$ by construction.
			
		\end{itemize}
		
\end{proof}
%

\emph{Let us prove Proposition~\ref{prop:completeness}}.
\begin{proof}
	Let us define fours binary relations:
	\begin{itemize}
		\item ${\mathcal P} \defn \{ ( \confCPS{\funEdit{P}}{P} , \confCPS{\funEdit{P}}{P | M}) \, \mid  \textrm{ for any $P$ and $M$} \} $; 
		\item ${\mathcal S} \defn \{ ( \confCPS{\funEdit{S}}{S} , \confCPS{\funEdit{S}}{S | M}) \, \mid  \textrm{ for any $S$ and $M$} \} $; 
		\item ${\mathcal C} \defn \{ (\confCPS{\funEdit{C}}{C}, \confCPS{\funEdit{C}}{C |  M} ) \, \mid  \textrm{ for any $C$ and $M$} \} $; 
		\item ${\mathcal A} \defn \{ (  \confCPS{\funEdit{A}}{A} , \confCPS{\funEdit{A}}{A | M} ) \, \mid  \textrm{ for any $A$ and $M$} \} $. 
	\end{itemize}
	It is enough to prove that the relation $\rel \defn {\mathcal P} \cup {\mathcal S} \cup {\mathcal C} \cup {\mathcal A}$ is a \emph{weak simulation}. For each pair $( N_1,  N_2 ) \in \rel$, the proof proceeds by case analysis on why $N_1 \trans{\alpha} N_1'$. 
	The proofs relies on two main facts: 
	(i) in a compromised controller $Z | M$ the malware $M$ cannot prevent the execution of actions that $Z$ would execute in isolation, in particular, 
	(ii)  $M$ cannot prevent the passage of the time of  $Z$, as recursion processes in $M$ are always time guarded. 
\end{proof}
%
{
Let us prove Proposition~\ref{cor:false-negatives} (Anomaly detection)
\setcounter{corollary}{1}
\begin{proof}
	As $P \trans{t }Z$, for some $Z$, by an application of transition rules \rulename{Go} and 
	\rulename{Enforce} it follows that $\confCPS{\go}{P} \trans{t } \confCPS{\go}{Z}$. By an application of the implication from right to left of Theorem~\ref{thm:w-enf-ctrl} it follows that 
	$ \confCPS{\funEdit{P}}{ P | M}  \trans{t } N$, for some $N$. 
	Thus, by Definition~\ref{def:detection}, 
	we have to show that $ \confCPS{\funEdit{P}}{ P | M}  \trans{t  \alpha} N$ for no  $N $.
	As by hypothesis $P \trans{ t \alpha} Z$ for no $Z$, it follows that  $\confCPS{\go}{P} \trans{ t\cdot \alpha} N  $,
	for no $N$. Thus, by an application of Theorem~\ref{thm:w-enf-ctrl}, this time from left to right, it follows that 
	$\confCPS{\funEdit{P}}{P | M} \trans{ t \alpha} N  $, for no $N$,
	as required. 
\end{proof} 
}

\emph{Let us prove Proposition~\ref{cor:weak-enforcement} (Weak enforcement of field networks)}.
\begin{proof}
	The result cannot be directly derived by an application of Theorem~\ref{thm:w-enf-ctrl} because in our timed setting  trace equivalence $\simeq$ is not preserved by parallel composition (the problem being the negative premise in the transition rule \rulename{TimeSync}). However, since the  weak simulation $\sqsubseteq$ is notoriously preserved by parallel composition, by an application of  Proposition~\ref{prop:transparency} and Proposition~\ref{prop:safety} we can easily derive: 
	\[
	\prod_{i=1}^n \confCPS{\funEdit{P_i}}{P_i | M_i} \; \sqsubseteq \; \prod_{i=1}^n \confCPS{\go}{P_i}   \, . 
	\]
	Furthermore, by an application of Proposition~\ref{prop:transparency} and  Proposition~\ref{prop:completeness} we can derive: 
	\[
	\prod_{i=1}^n \confCPS{\funEdit{P_i}}{P_i | M_i} \; \sqsupseteq \; \prod_{i=1}^n \confCPS{\go}{P_i}   \, . 
	\]
	This is enough to derive that the two systems under investigation are trace equivalent. 
\end{proof}

%

%
\emph{Let us prove Theorem~\ref{thm:s-enf-ctrl} (Observational enforcement of controllers)}.
\begin{proof}
	Let us define four binary relations:
	\begin{itemize}
		\item ${\mathcal S} \defn \{ ( \confCPS{\transf{S}}{S | M}, \confCPS{\transf{S}}{S}) \, \mid  \textrm{ for any $S$ and $M$} \} $; 
		\item ${\mathcal S} \defn \{ ( \confCPS{\transf{S}}{J}, \confCPS{\transf{S}}{S}) \, \mid  \textrm{ for any $S$ and $M$} \} $;
		\item ${\mathcal C} \defn \{ (\confCPS{\transf{C}}{J} ,  \confCPS{\transf{C}}{C}) \, \mid  \textrm{ for any $C$ and $J$} \} $; 
		\item ${\mathcal A} \defn \{ ( \confCPS{\transf{A}}{J} ,  \confCPS{\transf{A}}{A} ) \, \mid  \textrm{ for any $A$ and $J$} \} $. 
	\end{itemize}
	We prove that the relation $\rel \defn {\mathcal P} \cup {\mathcal S} \cup {\mathcal C} \cup {\mathcal A}$ is a weak bisimulation.  
	
	For each pair $( N_1 ,  N_2 ) \in \rel$ we proceed by case analysis on why $N_1 \trans{\alpha} N_1'$. Then, we do the same for $N_2 \trans{\alpha} N_2'$. We consider the most significant cases.

	\noindent
	{\bf Case} $(N_1,N_2) \in {\mathcal C}$. We do case analysis on the structure of the controller $C$ in $N_1= \confCPS{\transf{C}}{J}$, for some arbitrary $J$. 

\noindent Let $C \equiv \timeout{\sum_{i} c_i.C_i}{\hat{C}}$. 
		\begin{itemize}
			\item Let $N_1 \trans{c_i} \confCPS{\transf{C_i}}{J'} = N'$, by an application rule \rulename{Enforce} triggered by \rulename{Ctrl} or \rulename{Inject}, alternatively, by an application of rule \rulename{Mitigation}. Then, $N_2=\confCPS{\transf{C}}{C} \trans{c_i} \confCPS{\transf{C_i}}{C_i} = N_2'$ and $(N_1',N_2') \in {\mathcal C} \subset \rel$ by construction. 
			
			\item Let $N_1 \trans{\tick} \confCPS{\transf{\hat{C}}}{J'} = N'$, by an application rule \rulename{Enforce} triggered by an application of rule \rulename{TimeSync}, alternatively, by an application of rule \rulename{Mitigation}. Then, $N_2=\confCPS{\transf{C}}{C} \trans{\tick} \confCPS{\transf{\hat{C}}}{\hat{C}} = N_2'$ and $(N_1',N_2') \in {\mathcal C} \subset \rel$ by construction.
			
			\item  Let $N_1 \trans{\tau} \confCPS{\transf{C}}{J'} = N_1'$, by an application rule \rulename{Enforce} triggered by an application of rule \rulename{Inject} or \rulename{Ctrl}, because $J \trans{\alpha} J'$ for some $\alpha \in \ActSet^{\ast}\cup \ChanSet^{\ast} \setminus  \bigcup_i c_i $. Note that the edit automaton $\transf{C}$ \emph{suppresses} all  possible injections originating from the malware or the controller not aligned with the edit automaton, turning them into $\tau$-actions. 
			Note also that if $J \trans{a} J'$, for some actuator name $a$, then no drop is actually possible.
			Thus, $N_2=\confCPS{\transf{C}}{C} \Trans{\hat{\tau}} \confCPS{\transf{C}}{C} = N_2'$ and $(N_1,N_2') \in {\mathcal C} \subset \rel$ by construction.
		\end{itemize}
		Now, we proceed by case analysis on why $N_2 \trans{\alpha} N_2'$.
		\begin{itemize}
			\item Let $N_2=\confCPS{\transf{C}}{C} \trans{c_i} \confCPS{\transf{C_i}}{C_i} = N_2'$. Then, by definition of $\transf{C}$ we have $\transf{C} \trans{\eact{c_i}{c_i}} \transf{C_i}$ by an application of rule \rulename{Enforce} because $J \trans{c_i} J'$, it follows that $N_1 \trans{c_i} \confCPS{\transf{C_i}}{J'} = N_1'$, and $(N_1', N_2') \in {\mathcal C} \subset \rel$ by construction.  
			Note that, if $J \trans{\fineC}$, then by an application rule \rulename{Mitigation}, it follows that $N_1 \trans{c_i} \confCPS{\transf{C_i}}{J} = N_1'$, and $(N_1', N_2') \in {\mathcal C} \subset \rel$ by construction. 
			
			\item Let $N_2=\confCPS{\transf{C}}{C} \trans{\tick} \confCPS{\transf{\hat{C}}}{\hat{C}} = N_2'$. Then, by definition of $\transf{C}$ we have $\transf{C} \trans{\eact{\tick}{\tick}} \transf{\hat{C}}$ and $\transf{C} \trans{\eact{\alpha}{\tau}} \transf{C}$, for any action $\alpha \in\ActSet^{\ast}\cup \ChanSet^{\ast} \setminus  \bigcup_i c_i$ performed (injected) by the malware or the controller not aligned with the edit automaton. We recall that recursion in both malware and controller code is always \emph{time-guarded}, \emph{i.e.} may not prevent the passage time. More formally, for any $J$ there is a finite $n$ providing an upper bound to the maximum number of possible consecutive untimed actions of $J$.  Thus, by  $n$ possible  applications of rule \rulename{Inject} or \rulename{Ctrl} and a final application of  rule \rulename{TimeSync}, we have:    {\small $J  \trans{\alpha_1} J_1 \trans{\alpha_2} \cdots  \trans{\alpha_n} J_n \trans{\tick} J'$}.	
			By $n+1$ applications of rule \rulename{Enforce} we get $N_1 \Trans{\tick} \confCPS{\transf{\hat{C}}}{J'} = N_1'$, with $(N_1', N_2') \in {\mathcal C}  \subset \rel$ by construction. 
		\end{itemize}

\noindent Let $C \equiv \timeout{\overline{c}.C_1}{C_2}$.  This case is similar to the previous one. 

\noindent Let $C \equiv A$.  In this case, we resort to one of the next cases.

	\noindent
	{\bf Case} $(N_1,N_2) \in {\mathcal A}$. We proceed by  case analysis on the structure of the controller $A$ in $N_1= \confCPS{\transf{A}}{J}$, for some arbitrary $J$. 

\noindent Let $A \equiv \overline{a}.\hat{A}$. 
		\begin{itemize}
			\item Let $N_1 \trans{\overline{a}} \confCPS{\transf{\hat{A}}}{J'} = N'$, by an application rule \rulename{Enforce} triggered by \rulename{Ctrl} or \rulename{Inject}, alternatively, by an application of rule \rulename{Mitigation}. Then, $N_2=\confCPS{\transf{A}}{A} \trans{\overline{a}} \confCPS{\transf{\hat{A}}}{\hat{A}} = N_2'$ and $(N_1',N_2') \in {\mathcal A} \subset \rel$ by construction.
			
			\item  Let $N_1 \trans{\tau} \confCPS{\transf{A}}{J'} = N_1'$, by an application of \rulename{Enforce} triggered by an application of \rulename{Inject} or \rulename{Ctrl}, because $J \trans{\alpha} J'$ for some $\alpha \in \{\tau\} \,\cup \,\ActSet^{\ast}\cup \ChanSet^{\ast} \setminus  \{\overline{a}\}$. Note that the edit automaton $\transf{A}$ \emph{suppresses} all  possible injections originating from the malware or the controller not aligned with the edit automaton, turning them into $\tau$-actions. Furthermore, the misalignment may also occur as the malware has dropped the current actuation \rulename{DropAct}, thus the controller  performing an actuation $\alpha \in \overline{\ActSet} \cup \setminus  \{\overline{a}\}$ will be suppressed.
			Thus, $N_2=\confCPS{\transf{A}}{A} \Trans{\hat{\tau}} \confCPS{\transf{A}}{A} = N_2'$ and $(N_1,N_2') \in {\mathcal A} \subset \rel$ by construction.
		\end{itemize}
		Now, we proceed by case analysis on why $N_2 \trans{\alpha} N_2'$.
		\begin{itemize}
			\item Let $N_2=\confCPS{\transf{A}}{A} \trans{\overline{a}} \confCPS{\transf{\hat{A}}}{\hat{A}} = N_2'$. Then, by definition of $\transf{A}$ we have $\transf{A} \trans{\eact{\overline{a}}{\overline{a}}} \transf{\hat{A}}$ by an application of rule \rulename{Enforce} because $J \trans{\overline{a}} J'$, it follows that $N_1 \trans{\overline{a}} \confCPS{\transf{\hat{A}}}{J'} = N_1'$, and $(N_1', N_2') \in {\mathcal A} \subset \rel$ by construction.  
			Note that, if $J \trans{\fineC}$ by an application of rule \rulename{Mitigation}, it follows that $N_1 \trans{\overline{a}} \confCPS{\transf{\hat{A}}}{J} = N_1'$, and $(N_1', N_2') \in {\mathcal As} \subset \rel$ by construction. 
		\end{itemize}
		
\noindent Let $A \equiv \fineC.\Xrec$.	 
		\begin{itemize}
			\item Let $N_1 \trans{\fineC} \confCPS{\transf{S}}{J'} = N'$, by an application rule \rulename{Enforce} triggered by \rulename{Ctrl}. Note that by definition of the transition rule \rulename{Rec} we end up from $\fineC.\Xrec$ to the case  $\fineC.S$. Then, $N_2=\confCPS{\transf{A}}{A} \trans{\fineC} \confCPS{\transf{S}}{S} = N_2'$ and $(N_1',N_2') \in {\mathcal S} \subset \rel$ by construction.
			
			\item  Let $N_1 \trans{\tau} \confCPS{\transf{A}}{J'} = N_1'$, by an application rule \rulename{Enforce} triggered by an application of rule \rulename{Inject} or \rulename{Ctrl}, because $J \trans{\alpha} J'$ for some $\alpha \in \ActSet^{\ast} \cup \ChanSet^{\ast}$.   Here, notice that the edit automaton $\transf{A}$ \emph{suppresses} all  possible injections originating from the malware or the controller not aligned with the edit, turning them into $\tau$-actions. 
			Thus, $N_2=\confCPS{\transf{A}}{A} \Trans{\hat{\tau}} \confCPS{\transf{A}}{A} = N_2'$ and $(N_1,N_2') \in {\mathcal A} \subset \rel$ by construction.
		\end{itemize}
		Now, we proceed by case analysis on why $N_2 \trans{\alpha} N_2'$.
		\begin{itemize}
			\item Let $N_2=\confCPS{\transf{A}}{A} \trans{\fineC} \confCPS{\transf{S}}{S} = N_2'$. Here notice that by definition of the transition rule \rulename{Rec} we end up from $\fineC.\Xrec$ to the case  $\fineC.S$. Then, by definition of $\transf{A}$ we have $\transf{A} \trans{\eact{\fineC}{\fineC}} \transf{S}$ 
			and $\transf{A} \trans{\eact{\alpha}{\tau}} \transf{A}$, for any action $\alpha \in \ActSet^{\ast} \cup \ChanSet^{\ast}$ performed (injected) by the malware or the controller not aligned with the edit. Recall that the recursion in both malware and controller code is always \emph{time-guarded}, \emph{i.e.} may not prevent the passage time. Thus, the controller can always perform $\fineC$. More formally, for any $J$ there is a finite integer $n$ providing an upper bound to the maximum number of possible consecutive untimed actions of $J$.  Thus, by  $n$ possible  applications of rule \rulename{Inject} or \rulename{Ctrl} and a final application of  rule \rulename{Enforce}, we have:  {\small $J  \trans{\alpha_1} J_1 \trans{\alpha_2} \cdots  \trans{\alpha_n} J_n \trans{\fineC} J'$}. By $n+1$ applications of rule \rulename{Enforce} we get $N_1 \Trans{\fineC} \confCPS{\transf{S}}{J'} = N_1'$, with $(N_1', N_2') \in {\mathcal S}  \subset \rel$ by construction. 
		\end{itemize}
\end{proof}

	Before proving Proposition~\ref{cor:false-negatives-2} we need the following technical result, saying that controllers never deadlock.

	\begin{lemma}
		\label{lem:false-negatives-2}
		For any closed $Z \in \mathbbm{Ctrl} \cup \mathbbm{Sens} \cup \mathbbm{Comm} \cup \mathbbm{Act} $, it holds that
		$ Z  \trans{ \beta } \hat{Z}  $, for some $ \beta$ and $\hat{Z}$. 
	\end{lemma}
	\begin{proof}
		The proof is by induction on the structure of $Z$.
		\begin{itemize}
			\item 
			If $Z \equiv \fix \Xrec S$ then, by definition of the transition rule \rulename{Rec}, we end up to one of the other cases. 
			\item  If $Z \equiv \timeout{\sum_{i} s_i.S_i}{\hat{S}}$ then we can apply two different transitions rules:
			\begin{itemize}
				\item by an application of  rule \rulename{ReadS} we have $\hat{Z} \trans{s_j}S_j $, for some $j \in I$; 
				\item by an application of rule \rulename{TimeoutS} we have  $\hat{Z} \trans{\tick}S  $. 
			\end{itemize}
			In both cases, we have  $   \timeout{\sum_{i} s_i.S_i}{\hat{S}} \trans{ \beta } \hat{Z}  $, for some $ \beta$ and $\hat{Z}$.  
			\item 
			The other cases can be proved in a similar manner since $Z$ is closed, and in particular $Z\neq \fineC.\Xrec$.
		\end{itemize}
	\end{proof}

	\emph{Let us prove Proposition~\ref{cor:false-negatives-2} (Anomaly detection and mitigation)}.
	\begin{proof}
		Let us prove the first item of the proposition.
		As there is a controller $Z$ such that $P \trans{t }Z$, by an application of transition rules \rulename{Go} and 
		\rulename{Enforce} it follows that $\confCPS{\go}{P} \trans{t } \confCPS{\go}{Z}$. By an application of the implication from right to left of  Theorem~\ref{thm:s-enf-ctrl} it follows that 
		$ \confCPS{\funEdit{P}}{ P | M}  \trans{t } N$, for some $N$. 
		Thus, by Definition~\ref{def:detection}, 
		we have to show that $ \confCPS{\funEdit{P}}{ P | M}  \trans{t  \alpha} N$ for no  $N $.
		As by hypothesis $P \trans{ t \alpha} Z$ for no $Z$, it follows that  $\confCPS{\go}{P} \trans{ t\cdot \alpha} N  $,
		for no $N$. Thus, by an application of Theorem~\ref{thm:s-enf-ctrl}, this time from left to right, it follows that 
		$\confCPS{\funEdit{P}}{P | M} \trans{ t \alpha} N  $, for no $N$,
		as required. 
		
		Let us prove the second item of the proposition. We proceed by contradiction, showing  that 
		$\confCPS{\transf{P}}{P | M} \trans{ t \beta} N  $, for some $\beta \neq \alpha$.
		Suppose $\confCPS{\transf{P}}{P | M} \trans{ t \beta} N  $, for no  $ \beta$ and $N$.
		By an application Theorem~\ref{thm:s-enf-ctrl}, from right to left,   it follows that 
		$ \confCPS{\go}{P}   \trans{ t \beta} N  $, for no $ \beta$ and $N$. Thus,  
		$ P  \trans{ t \beta} Z  $, for no $ \beta$ and $Z$. Since, by hypothesis,  $ P  \trans{ t } Z  $, for some $Z$,
		it holds that $ Z  \trans{ \beta } \hat{Z}  $, for no $ \beta$ and $\hat{Z}$. 
		This is in contradiction with Lemma \ref{lem:false-negatives-2}.
		As a consequence, it must be that $\confCPS{\transf{P}}{P | M} \trans{ t \beta} N  $ for some $\beta \neq \alpha$.
		Now, there are two possibilities:
		\begin{itemize}
			\item either $\beta$ is derived by an application of rule \rulename{Mitigation}, with 
			  $\alpha=\fineC$ and  $\transf{P}\trans{\eact{\beta}{\beta}}$, and by definition of our synthesis function, we have that  $\tau \neq \beta \neq \alpha=\fineC$;  
			\item or $\beta$ is derived by an application of rule \rulename{Enforce} and, 
			by definition our synthesis,  the monitor $\transf{P}$   suppresses the action $\alpha$,
			namely, $\alpha \neq \tau$ and $\beta=\tau$. 
		\end{itemize}
	\end{proof}

\end{document}